\documentclass[journal]{IEEEtran}
\IEEEoverridecommandlockouts

\usepackage{amssymb}
\usepackage[cmex10]{amsmath}
\usepackage{stfloats}
\usepackage{graphicx}
\usepackage{subfigure}
\usepackage{tabularx}
\usepackage{epsfig,epsf,color,balance,cite}
\usepackage{verbatim}
\usepackage{url}
\usepackage{bm}
\usepackage{booktabs}
\usepackage{float}
\usepackage{subfigure}
\usepackage{subcaption}
\usepackage{caption}
\usepackage{amsmath}
\allowdisplaybreaks[2]
\usepackage[ruled,linesnumbered]{algorithm2e}

\usepackage{algorithmic}
\newtheorem{theorem}{\bf Theorem}
\newtheorem{lemma}{\bf Lemma}

\newtheorem{proof}{Proof}
\usepackage{color}
\definecolor{myc1}{rgb}{0,0,0}

\begin{document}

\title{ 
Energy Efficient Fluid Antenna Relay (FAR)-Assisted Wireless Communications
}

\author{
\IEEEauthorblockN{Ruopeng Xu,
                  Zhaohui Yang~\IEEEmembership{Member,~IEEE},
                  Zhaoyang Zhang~\IEEEmembership{Senior Member,~IEEE},\\
                  Mohammad Shikh-Bahaei~\IEEEmembership{Senior Member,~IEEE},
                  Kaibin Huang~\IEEEmembership{Fellow,~IEEE}, and Dusit Niyato~\IEEEmembership{Fellow,~IEEE}
                 } 
\thanks{R. Xu, Z. Yang, and Z. Zhang are with the College of Information Science and Electronic Engineering, Zhejiang University, and also with Zhejiang Provincial Key Laboratory of Info. Proc., Commun. \& Netw. (IPCAN), Hangzhou, 310027, China (e-mails: \{ruopengxu, yang\_zhaohui, ning\_ming\}@zju.edu.cn).}
\thanks{Mohammad Shikh-Bahaei is with the Department of Engineering, King’s College London, London, UK (email: m.sbahaei@kcl.ac.uk).}
\thanks{Kaibin Huang is with the Department of Electrical and Electronic Engineering, The University of Hong Kong, Hong Kong SAR (email: huangkb@eee.hku.hk).}
\thanks{Dusit Niyato is with the School of Computer Engineering, Nanyang Technological University, Singapore(email: dniyato@ntu.edu.sg).}

\vspace{-2em}
}

\maketitle
\begin{abstract}
In this paper, we propose an energy efficient wireless communication system based on fluid antenna relay (FAR) to solve the problem of non-line-of-sight (NLoS) links caused by blockages with considering the physical properties. Driven by the demand for the sixth generation (6G) communication, fluid antenna systems (FASs) have become a key technology due to their flexibility in dynamically adjusting antenna positions. Existing research on FAS primarily focuses on line-of-sight (LoS) communication scenarios, and neglects the situations where only NLoS links exist. To address the issues posted by NLoS communication, we design an FAR-assisted communication system combined with amplify-and-forward (AF) protocol. In order to alleviate the high energy consumption introduced by AF protocol while ensuring communication quality, we formulate an energy efficiency (EE) maximization problem. By optimizing the positions of the fluid antennas (FAs) on both sides of the FAR, we achieve controllable phase shifts of the signals transmitting through the blockage which causes the NLoS link. Besides, we establish a channel model that jointly considers the blockage-through matrix, large-scale fading, and small-scale fading. To maximize the EE of the system, we jointly optimize the FAR position, FA positions, power control, and beamforming design under given constraints, and propose an iterative algorithm to solve this formulated optimization problem. Simulation results show that the proposed algorithm outperforms the traditional schemes in terms of EE, achieving up to $23.39\%$ and $39.94\%$ higher EE than the conventional reconfigurable intelligent surface (RIS) scheme and traditional AF relay scheme, respectively.
\end{abstract}

\IEEEpeerreviewmaketitle

\section{Introduction}\label{Introduction}
Driven by the fast development of the sixth generation (6G) wireless communication and its high demand service requirements, fluid antenna system (FAS) has recently emerged as a promising technology and has drawn plenty of attention from both academia and industry. Unlike traditional antenna systems with fixed physical positions, FAS allows fluid antennas (FAs) to flexibly and instantly adjust their locations in the discrete or continuous space within a predefined region. Authors in \cite{wong2020fluid} first introduce the concept of FAS, proposing the exact and approximate closed-form expressions and upper bound for the outage probability of FAS. Soon afterwards, research on both theoretical basis \cite{wong2020performance,chai2022port,psomas2023continuous}, and physical architectures\cite{abu2021liquid,zhang2024pixel,ghadi2024physical}, has flourished, laying foundations of the development of FAS. 

Advancements in fields, such as extra-large multiple-input multiple-output (MIMO)\cite{wang2024fluid} and artificial intelligence (AI)\cite{wang2024ai}, have promoted the study of FAS. At the same time, 
research on FAS, as well, has constantly boosted the development of other fields in wireless communications, such as multiple access technology\cite{wong2021fluid}, channel estimation\cite{skouroumounis2022fluid}, integrated sensing and communication (ISAC)\cite{zhou2024fluid}, MIMO system\cite{ye2023fluid}, near-field communication\cite{10767351}, and secret communication\cite{10092780}. The work in \cite{wong2021fluid} explores the possibility for utilizing a single FA of a mobile user for multiple access, i.e., fluid antenna multiple access (FAMA). In FAS without perfect channel state information (CSI), the authors in \cite{skouroumounis2022fluid} study the trade-off between system performance and channel estimation quality. In different communication systems, the researchers in \cite{zhou2024fluid} utilize FAs to enhance the performance of ISAC system, and the work in \cite{ye2023fluid} has designed the FA-assisted MIMO system with statistical CSI. In terms of different communication scenarios, authors in \cite{10767351}
investigate a point-to-point near-field communication with FAS. In the system proposed in \cite{10092780}, secrecy rate of the user with a single FA is comparable to that of the user utilizing maximum ratio combining (MRC) with multiple antennas. 

However, existing research on FAS primarily focuses on scenarios with line-of-sight (LoS) links but neglects the challenges posed by blockages that result in non-LoS (NLoS) links. This leaves a research gap in FAS-aided wireless communication systems that limits FAS’s applicability in practical environments with blockages. Auxiliary communication devices for assistance, such as a reconfigurable intelligent surface (RIS)\cite{gan2021ris,magbool2024surveyintegratedsensingcommunication,10720877}, a simultaneously transmitting and reflecting (STAR)-RIS\cite{10133841,10380743,10685065}, and an autonomous aerial vehicle (AAV)\cite{abdou2024sum} have already been explored for relieving NLoS challenges. However, RISs can only reflect signals and thus cannot solve NLoS issues where LoS links cannot be established via reflection alone. For example, if the receiver is surrounded by blockages requiring the signal to penetrate them, RISs cannot establish the LoS link. On the other hand, the system models in STAR-RIS research primarily assume that signals only need to pass through STAR-RIS to reach the users on the other side of the blockage, ignoring the impact of obstacles that deploy STAR-RIS inside on signal transmission. In particular, different widths or different media of the obstacles may cause different amplitude attenuation and phase shift of signals. The study in \cite{abdou2024sum} considers the scenario where an AAV is introduced as a relay with decode-and-forward (DF) protocol to receive and retransmit signals, in which the assumption with DF protocol puts high requirements to the relay network and is unfriendly with energy consumption. 

These challenges motivate our design of a fluid antenna relay (FAR) that dynamically optimizes antenna positions on both sides of a blockage, enabling efficient NLoS signal transmission with the help of a specially designed medium within the blockage. To clarify the practical deployment of the FAR system, we primarily focus on the two parts: 1) the practical deployment of the FAS on both sides of the blockage, and 2) the embedding process of the isotropic medium within the blockage. From this point of view, we note that the research on FAS hardware designs has been constantly developing\cite{10740058,9539785,9977471,7950976}, providing the practical foundation for the deployment of FAs on the surfaces of the blockage. On the other hand, research works\cite{10907789,11059266,10930784,10707271} have considered scenarios where communication devices are embedded in the blockages, such as walls or doors, to help signal transmission. Hence, the embedding process of the isotropic medium within the blockage can be achieved adopting a similar embedded technology. Using the amplify-and-forward (AF) protocol which is more hardware-friendly compared with DF protocol, FAR can receive and forward the signals. Besides, by adjusting the positions of FAs on the FAR, we can complete controllable phase shift of the transmission through the blockage. 

The key contributions of this paper include:
\begin{itemize}
    \item We first model an FAR with the AF protocol, and then derive the amplitude attenuation coefficients and phase shift of signals with the help of the FAR transmitting through the blockage. Utilizing Maxwell's equations, we prove that the phase shift of signals can be controllable with the change of positions of FAs on both sides of the FAR. On this basis, the blockage-through matrix $\bm{\Theta}$ is modeled and then the channel model jointly considering the blockage-through matrix, large-scale fading influenced by the positions of reference points on both sides of the FAR, and small-scale fading caused by the phase difference of FAs is established.
    \item We investigate an uplink wireless communication system with the FAR that can receive the signals from users and transmit them to the BS with a controllable signal phase shift. To maximize the energy efficiency (EE) of the proposed FAR-assisted system, we jointly optimize the position of reference points, which locates the predefined region for FAs to switch their positions, the positions of the FAs, power control of users, and beamforming design at the BS.  On this basis, we formulate an optimization problem with the objective to maximize the EE under the minimum rate constraint of users, maximum power constraint of users, and position constraints of the FAs.
    \item To solve the EE maximization problem for multiples users, a sub-optimal solution is obtained by an iterative algorithm through solving the large-scale fading optimization sub-problem, the small-scale fading optimization sub-problem, and joint power control and beamforming design sub-problem. For the large-scale fading optimization sub-problem, we first derive the closed form of one reference point of the FAR, and a sub-optimal solution to the other reference point is obtained by the successive convex approximation (SCA) method. For the small-scale fading optimization sub-problem, we propose an alternating optimization (AO) method to alternately obtain the sub-optimal position of the FAs at users, the FAR, and the BS. In terms of joint power control and beamforming design sub-problem, we introduce a Dinkelbach method based algorithm to solve it.
\end{itemize}
Simulation results show that our proposed system outperforms the STAR-RIS scheme and the traditional AF relay scheme, improving up to $23.39\%$ and $39.94\%$ EE, respectively.

\textit{Notation}: Unless otherwise explicitly stated, we use lower case letters to denote scalars, bold lower case letters to denote vectors, and bold upper letters to denote matrices. $j = \sqrt{-1}$ is the imaginary unit. The subscripts $(\cdot)^{*}$, $(\cdot)^{T}$ and $(\cdot)^{H}$ denote the complex conjugate, transpose and conjugate transpose (Hermitian) operations, respectively. The operation $\mathrm{diag}(\mathbf{a})$ generates a diagonal matrix with the elements of $\mathbf{a}$ along its main diagonal, the operation $\mathrm{Re}(a)$ obtains the real part of scaler $a$, 
and the operation $\mathrm{tr}(\mathbf{A})$ generates the trace of $\mathbf{A}$. $\mathrm{vec}(\mathbf{A})$ represents the column-wise vectorization of $\mathbf{A}$. $\mathcal{CN}(\mathbf{a},\mathbf{B})$ represents the symmetric complex-valued Gaussian distribution with mean $\mathbf{a}$ and covariance matrix $\mathbf{B}$. $||\mathbf{A}||_2$ is the spectral norm of $\mathbf{A}$, $||\mathbf{A}||_F$ is the Euclidean (Frobenius) norm of $\mathbf{A}$, $||\mathbf{a}||_2$ stands for the Euclidean norm of $\mathbf{a}$, and $|a|$ represents the modulus of $a$. In terms of operators, $\times$ is cross product operator, $\cdot$ is the dot product operator, and $\nabla$ is the gradient operator. $\mathbf{A}\succeq0$ means the $\mathbf{A}$ is semi-positive definite, and $\mathbf{A}\preceq0$ means the $\mathbf{A}$ is semi-negative definite. $\mathbf{I}_N$ denotes an identity matrix in $N \times N$ dimension. The sets of $M \times N$ dimensional complex and real matrices are denoted by $\mathbb{C}^{M \times N}$ and $\mathbb{R}^{M \times N}$, respectively.

\begin{figure}[t]
\centering
\includegraphics[width=1\linewidth]{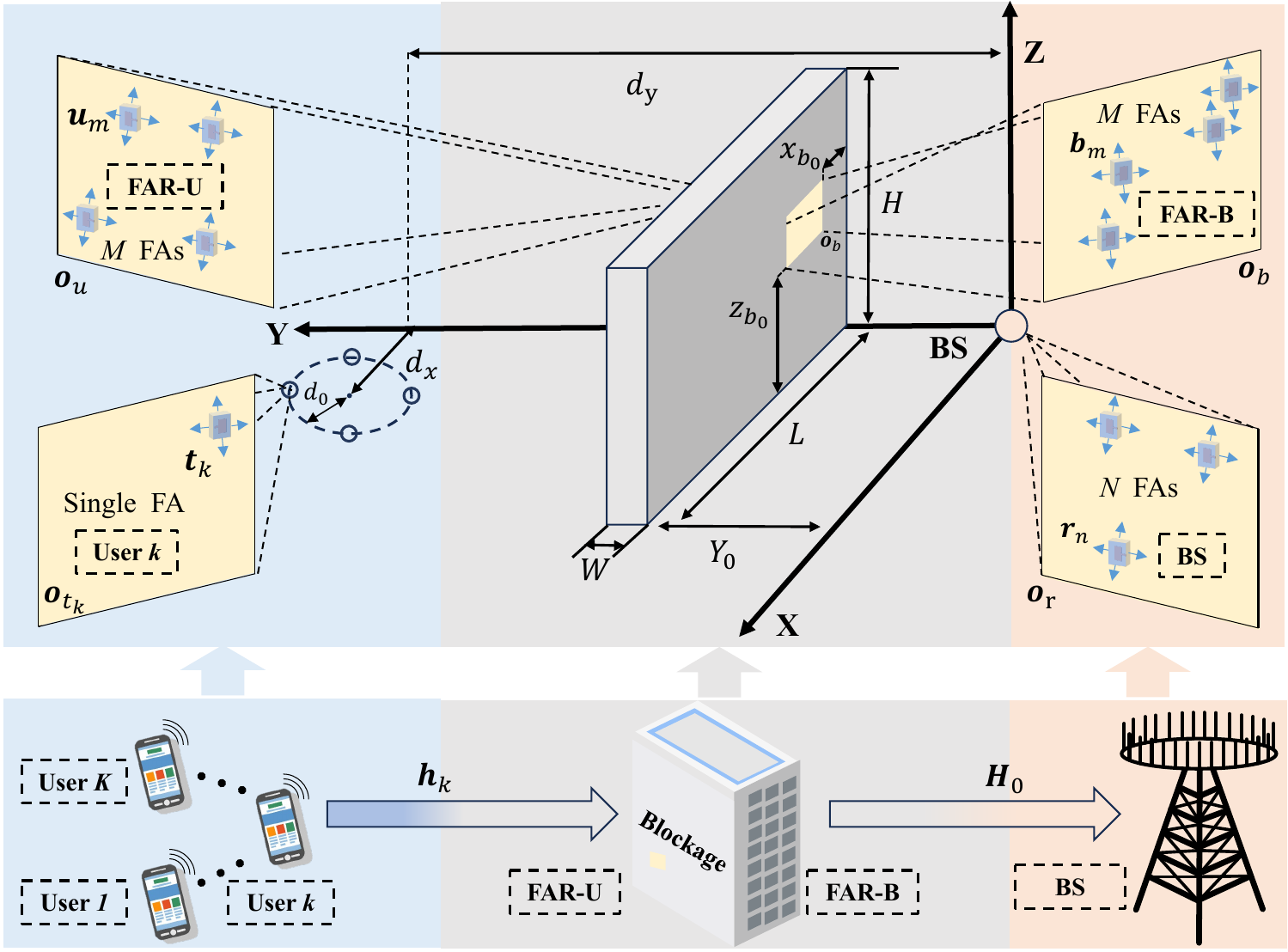}
\caption{System model of the FAR-assisted wireless communication.}
\label{SystemModel}
\end{figure}

\section{System Model}
As illustrated in Fig.~\ref{SystemModel}, we consider an upink FAR-assisted wireless communication network with one BS, $K$ users, and a blockage deployed FAR in the LoS link between the BS and the multiple users. In the considered model, every user has a single FA, the BS has $N$ FAs, $M$ FAs are deployed on each side of the blockage,  and all FAs are assumed to be connected to radio frequency chains via flexible cables, so that they can freely change their positions within a given region at real time\cite{ma2023mimo}. Without loss of generality, we assume that the blockage is a cubic structure (such as a wall, a window, or a door) with the length of $L$, the width of $W$, and the height of $H$. To describe spatial relations, we introduce the three-dimensional Cartesian coordinate system, where the locations of the BS and user $k$ are $\mathbf{o}_r = [0,0,0]^T$ and $\mathbf{o}_{k} = [x_{k_0},y_{k_0},0]^T$, respectively, and the distance between the FAR and BS along y-axis is $Y_0$.

In the following of this section, we first introduce the basic model of the FAR from oblique view, front view, and right view, and on this basis, we model the process how signals transmit through the blockage with the help of the FAR. Subsequently, we give the channel model from user $k$ to BS jointly considering small-scale fading and large-scale fading. Based on the channel model, the signal received at the BS and its corresponding signal-to-interference-plus-noise (SINR) are represented. Jointly considering the optimization with large-scale fading, small-scale fading, user power control, and beamforming design, we formulate the problem to maximize the EE of the proposed system.

\begin{figure}[t]
\centering
\includegraphics[width=1\linewidth]{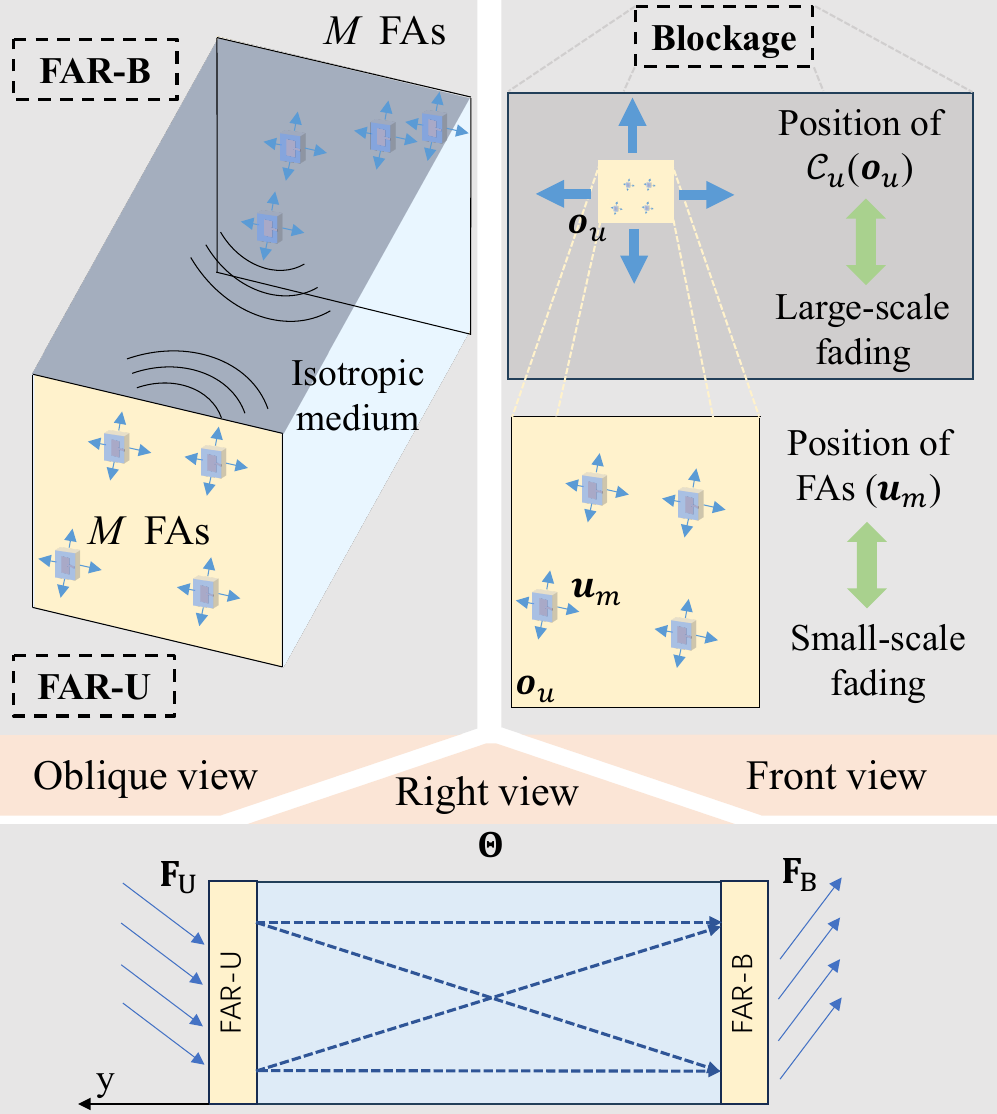}
\caption{An illustration of the FAR model.}
\label{FARModel}
\end{figure}

\subsection{Fluid Antenna Relay Model}
As shown in the oblique view part of Fig.~\ref{FARModel}, the FAR architecture integrates two key components: 1) FAs on each side of the blockage, capable of dynamically adjusting their positions within predefined regions, and 2) an isotropic medium embedded within the blockage to enable signal transmission. In particular, in the considered uplink transmission, $M$ FAs on the FAR receiving signals from the users are denoted as an FAR-U (FAR-User side), and $M$ FAs on the other side are denoted as an FAR-B (FAR-BS side), which are connected with the isotropic medium embedded within the blockage.

The regions $\mathcal{C}_u$ and $\mathcal{C}_b$, each of size $C \times C$, define the movable ranges for the FAR-U and the FAR-B, respectively. Set the FAR-U as an example, as revealed in the front view part of Fig.~\ref{FARModel}, FAs are deployed at the surface of the blockage and can move freely within $\mathcal{C}_u$. Let $\mathbf{o}_u = [x_{u_0},Y_0+W,z_{u_0}]^T$ denote the reference point of $\mathcal{C}_u$,  and $\mathcal{C}_u$ can be mathematically denoted by $\{x,y,z|,x\in[x_{u_0},x_{u_0}+C],y = Y_0+W, z \in [z_{u_0},z_{u_0}+C]\}$. Accordingly, the three-dimensional location of the $m$-th FA of the FAR-U can be denoted by $\mathbf{u}_m = [x_{u_m},y_{u_m},z_{u_m}]^T$, i.e., $\mathbf{u}_m \in \mathcal{C}_u$. Since the size of the blockage surface is comparable with the signal propagation distance, the position of $\mathcal{C}_u$ influences the large-scale fading of signal transmission. When $\mathcal{C}_u$ is determined, FA position within $\mathcal{C}_u$ influences small-scale fading of signal transmission.  Similarly, letting $\mathbf{o}_b = [x_{b_0},Y_0,z_{b_0}]^T$ denote the reference point of $\mathcal{C}_b$, and $\mathbf{b}_m=[x_{b_m},y_{b_m},z_{b_m}]^T$ denote the position of the $m$-th FA of the FAR-B, we have $\mathcal{C}_b = \{x,y,z|,x\in[x_{b_0},x_{b_0}+C],y = Y_0, z\in[z_{b_0},z_{b_0}+C] \}$, and $\mathbf{b}_m \in \mathcal{C}_b$.

Moreover, as illustrated in the right view part of the FAR model, in an uplink transmission, the FAR-U receives signals from users and utilizes the AF protocol to transmit the signals to the FAR-B, where the forwarding matrix of the FAR-U is $\mathbf{F}_U = \mathrm{diag}(f_1^u,\dots,f_M^u)$. Similarly, the forwarding matrix of the FAR-B can be denoted as $\mathbf{F}_B = \mathrm{diag}(f_1^b,\dots, f_M^b)$. 

To model the process of signal transmitting through the blockage, we assume that the isotropic medium is with conductivity $\sigma$, permittivity $\varepsilon$, and magnetic permeability $\mu$.

\begin{theorem}
   During the transmission through the blockage, the amplitude of the signal from the $q$-th FA of the FAR-U and the $p$-th FA of the FAR-B attenuates with a coefficient $\alpha_{pq}$ and with a phase shift $e^{-j\theta_{pq}}$. They can be presented as 
    \begin{equation}
   \alpha_{pq} \approx e^{-C_1\mathrm{sec}^{\frac{1}{2}}(\gamma)\mathrm{sin}(\frac{1}{2}\gamma)||\mathbf{o}_b-\mathbf{o}_u||_2} \triangleq \alpha,     
    \end{equation}
  \begin{equation}\label{pahseShift}
              \theta_{pq} =  C_1\mathrm{sec}^{\frac{1}{2}}(\gamma)\mathrm{cos}(\frac{1}{2}\gamma)||\mathbf{b}_p -\mathbf{u}_q||_2,
  \end{equation}
    where $C_1 = \omega\sqrt{\mu\varepsilon}$, $\gamma = \mathrm{arctan}(\frac{\sigma} {\omega\varepsilon})$, $\omega$ is the angular frequency, $||\mathbf{o}_b-\mathbf{o}_u||_2$ is the distance between reference points of the FAR-U and the FAR-B, and $||\mathbf{b}_p-\mathbf{u}_q||_2$ is the distance between the $q$-th FA at FAR-U and the $p$-th FA at FAR-B.
\end{theorem}

\begin{proof}
    See Appendix~\ref{Proof1}.
\end{proof}

Without loss of generality, we set $\mathbf{F}_U = \mathbf{F}_B = \sqrt{\alpha^{-1}}\mathbf{I}_M$ as the forwarding matrices of the FAR-U and the FAR-B, respectively. The process of the signals transmitted by the FAR can be modeled by the blockage-through matrix $\bm{\Theta}$, describing the controllable phase shifts due to the transmission through the blockage as followings:
\begin{align}\label{FARMatrix}
    \begin{aligned}
      \bm{\Theta} 
      &=
      \mathbf{F}_B
      \begin{bmatrix}
      \alpha_{11}e^{-j\theta_{11}}  & \dots & \alpha_{1M}e^{-j\theta_{1M}} \\
      \vdots & &  \vdots\\
      \alpha_{M1}e^{-j\theta_{M1}}  & \dots & \alpha_{MM}e^{-j\theta_{MM}} \\
      \end{bmatrix}
      \mathbf{F}_U\\ 
      &=\begin{bmatrix}
      e^{-j\theta_{11}}  & \dots & e^{-j\theta_{1M}} \\
      \vdots & &  \vdots\\
      e^{-j\theta_{M1}}  & \dots & e^{-j\theta_{MM}} \\
      \end{bmatrix} \in \mathbb{C}^{M\times M}.
    \end{aligned}
\end{align}


\subsection{Channel Model}

Similarly to $\mathcal{C}_u$ and $\mathcal{C}_b$, the regions given for an FA at user $k$ and FAs of BS to be freely reconfigurable are denoted as $\mathcal{C}_{k}$ and $\mathcal{C}_r$, respectively, which are both in the size of $C \times C$.  All regions for FAs to be reconfigurable are assumed to be much smaller than the signal propagation distance\cite{zhou2024fluid,ye2023fluid,ma2023mimo}, and the far-filed model can be utilized. Thus, the channel is modeled with a planar wavefront\cite{zhu2023modeling}, and FAs in the same region have the same angle of departure (AoD)/angle of arrival (AoA) and the same amplitude. For user $k$, to precisely describe the phase difference of the FA in different positions, we consider the reference point of $\mathcal{C}_k$ to be $\mathbf{o}_{k}$, the position vector of the FA of user $k$ to be $\mathbf{t}_k = [x_{t_k}, y_{k_0}, z_{t_k}]^T$, and the elevation and azimuth of AoD can be denoted as $\theta_{t_k} \in [0,\pi]$ and $\phi_{t_k} \in [0,\pi]$, respectively. As a result, the phase difference between the FA at $\mathbf{t}_k$ and $\mathbf{o}_{k}$ can be given as
\begin{equation}\label{phaseDifference}
\rho(\mathbf{t}_k,\bm{\kappa}_{t_k},\mathbf{o}_{k}) 
    = e^{j \frac{2\pi}{\lambda}(\bm{\kappa}_{t_k})^T (\mathbf{t}_k - \mathbf{o}_{k})},
\end{equation}
where $\bm{\kappa}_{t_k} = [\mathrm{sin}\theta_{t_k}\mathrm{cos}\phi_{t_k},\mathrm{cos}\theta_{t_k},\mathrm{sin}\theta_{t_k}\mathrm{sin}\phi_{t_k}]^T$ is the wave vector of user $k$ and $\lambda $ is the carrier wavelength. Given that $\bm{\kappa}_{t_k}$ and $\mathbf{o}_{k}$ are generally determined, the phase difference between $\mathbf{o}_{k}$ and $\mathbf{t}_{k}$ can be denoted as $\rho(\mathbf{t}_k)$.

Similarly, we denote the phase difference between the $m$-th FA at the FAR-U and $\mathbf{o}_u$, the $m$-th FA at the FAR-B and $\mathbf{o}_b$, and the $n$th FA at BS and $\mathbf{o}_r$ by $\rho(\mathbf{u}_m)$, $\rho(\mathbf{b}_m)$, and $\rho(\mathbf{r}_n)$, respectively, where $\mathbf{r}_n=[x_{r_n},0,z_{r_n}]^T$ is the position vector of the $n$-th FA at the BS.

$\mathbf{U} = [\mathbf{u}_1,\dots,\mathbf{u}_M]$, $\mathbf{B} = [\mathbf{b}_1,\dots,\mathbf{b}_M]$, and $\mathbf{R} = [\mathbf{r}_1,\dots,\mathbf{r}_N]$ are the position matrices of FAs for the FAR-U, the FAR-B, and at BS, respectively. Then, the phase difference vectors of the FAR-U, the FAR-B, and BS can be represented by $\bm{\rho} (\mathbf{U}) = [\rho(\mathbf{u}_1),\dots,\rho(\mathbf{u}_M)]^T$, $\bm{\rho} (\mathbf{B}) = [\rho(\mathbf{b}_1),\dots,\rho(\mathbf{b}_M)]^T$, and $\bm{\rho} (\mathbf{R}) = [\rho(\mathbf{r}_1),\dots,\rho(\mathbf{r}_N)]^T$, respectively. In addition, we denote $\mathbf{T}=[\mathbf{t}_1,\dots,\mathbf{t}_K]$ as the position matrix to represent the FA positions of $K$ users.

In one specific transmission, the users first transmit the signals to the FAR. Then, the signals transmit through the blockage and are received by the BS. To jointly consider small-scale fading and large-scale fading, we use the Rician distribution to model the small-scale fading\cite{gan2021ris} and  use symbol $\sqrt{\beta_k}$ to denote large-scale fading from user $k$ to the FAR-U\cite{laourine2008capacity}, which is given as
\begin{equation}
   \beta_k = \frac{\mu_k}{D_k^l}, \forall k \in \{1,\dots,K\},
\end{equation}
where the large-scale fading coefficient $\mu_k$ is an independent and identically distributed (i.i.d.) Gamma random variable, $D_k = ||\mathbf{t}_k-\mathbf{o}_u||_2$ is the distance between user $k$ and the FAR-U, and $l$ is the path-loss exponent with typical values ranging from $2$ to $6$\cite{yang2015role}. On this basis, the channel from user $k$ to the FAR-U can be modeled as

\begin{equation}\label{hk}
    \sqrt{\beta_k}\mathbf{h}_k = \sqrt{\beta_k}\left(\sqrt{\frac{K_{1}}{K_1+1}}\tilde{\mathbf{h}}_k  + \sqrt{\frac{1}{K_1+1}}\hat{\mathbf{h}}_k\right) \in \mathbb{C}^{M\times1},
\end{equation}
where $K_1$ is the Rician factor, $\tilde{\mathbf{h}}_k=\bm{\rho}(\mathbf{U})\rho^H(\mathbf{t}_k)$ is the LoS component of $\mathbf{h}_k$, and $\hat{\mathbf{h}}_k$ is the NLoS component, whose elements are i.i.d. complex Gaussian random variables with the zero mean and unit variance.

Similarly, $\sqrt{\beta_0}$ represents the the large-scale fading from the FAR-B to BS, which can be given as,
\begin{equation}
    \sqrt{\beta_0} = \frac{\mu_0}{D_0^l},
\end{equation}
where $\mu_0$ is the i.i.d. Gamma random variable, and $D_0 = ||\mathbf{o}_b-\mathbf{o}_r||_2$ is the distance between the FAR-B and BS.

Then, the channel from the FAR-B to BS is 
\begin{equation}
    \sqrt{\beta_0}\mathbf{H}_0 = \sqrt{\beta_0}\left(\sqrt{\frac{K_{0}}{K_0+1}}\tilde{\mathbf{H}}_0  + \sqrt{\frac{1}{K_0+1}}\hat{\mathbf{H}}_0\right) \in \mathbb{C}^{N\times M},
\end{equation}
where $K_0$ is the Rician factor, $\tilde{\mathbf{H}}_0 = \bm{\rho}(\mathbf{R})\bm{\rho}^H(\mathbf{B}) $ is the LoS component of $\mathbf{H}_0$, and $\hat{\mathbf{H}}_0$ is the NLoS component  with complex Gaussian random distributed i.i.d. entries with zero means and unit variances.

\subsection{Communication Model}
The received signal at the BS can be given as
\begin{equation}
 \mathbf{y}_{BS}  = \sum_{k=1}^K \sqrt{p_k} \sqrt{\beta_k\beta_0} \underbrace{\mathbf{H}_0 \bm{\Theta}}_{\triangleq\mathbf{H}} \mathbf{h}_k s_k + \sqrt{\beta_0}\mathbf{H} \mathbf{n}_u + \mathbf{n}_r,
\end{equation}
where $p_k$ is the transmit power of user $k$, satisfying $p_k \leq P_{k,\mathrm{max}}$, and $s_k$ represents the information symbol for user $k$, which is modeled as $\mathcal{CN}(0,1)$. $\mathbf{n}_u \sim \mathcal{CN}(\mathbf{0},\sigma_u^2 \mathbf{I}_M)$ is the additive white Gaussian noise (AWGN) at the FAR-U, and $\mathbf{n}_r \sim \mathcal{CN}(\mathbf{0},\sigma_r^2 \mathbf{I}_M)$ is the AWGN at the BS.

The achievable rate of user $k$ can be given by $R_k =  \rm{log} (1 + \gamma_k)$, where $\gamma_k$ is the SINR, given as
\begin{equation}\label{SINR}
    \gamma_k = \dfrac{p_k\beta_0\beta_k|\bm{\omega}_k^H\mathbf{H} \mathbf{h}_k|^2}{\sigma_r^2+\beta_0\sigma_u^2||\bm{\omega}_k^H\mathbf{H}||_2^2  +\sum_{i=1,i\neq k}^{K} p_i\beta_0\beta_i|\bm{\omega}_k^H\mathbf{H} \bm{h}_i|^2}, 
\end{equation}
where $\bm{\omega}_k$ is the beamforming vector for user $k$.

\subsection{Problem Formulation}
Given the considered system model, we assume that the reference points of users and the BS are generally fixed. Then, our objective is to jointly optimize the reference points of the FAR, i.e., $\mathbf{o}_u$ and $\mathbf{o}_b$,  the FA-position matrices, i.e., $\mathbf{T}$, $\mathbf{U}$, $\mathbf{B}$, $\mathbf{R}$, transmit power vector $\mathbf{p}$, and beamforming design,  represented by $\bm{\Omega}$, so as to maximize the EE under the minimum achievable rate requirements of users, the constraints of FA positions, and the maximum power constraint of each user. Mathematically, the problem for the FAR-assisted uplink communication can be given by:
\begin{subequations}\label{sys1max0}
    \begin{align} 
       & \mathop{\max}_{\mathbf{o}_u, \mathbf{o}_b, \mathbf{T}, \mathbf{U}, \mathbf{B}, \mathbf{R},\mathbf{p}, \bm{\Omega}} \quad \dfrac{B\sum_{k=1}^K\mathrm{log}_2(1+\gamma_k)}{\sum_{k=1}^Kp_k + P_B + P_F} ,\tag{\ref{sys1max0}}\\
         \textrm{s.t.} \qquad 
         &B\mathrm{log}_2(1+\gamma_k) \geq R_{k,\mathrm{min}},\forall k \in \mathcal{K}, \\
   & 0\leq p_k\leq P_{k,\mathrm{max}},\forall k \in \mathcal{K},\\
   & \mathbf{t}_k \in \mathcal{C}_{k}, \forall k \in \mathcal{K}\\
   & \mathbf{u}_m \in \mathcal{C}_u, \forall m \in \mathcal{M}, \\
   & \mathbf{b}_m \in \mathcal{C}_b, \forall m \in \mathcal{M},  \\
   &  \mathbf{r}_{n}\in \mathcal{C}_r, \forall n \in \mathcal{N},\\
   & ||\mathbf{u}_{m}-\mathbf{u}_l||_2 \geq d_{\mathrm{min}}, \forall m,l \in \mathcal{M} , m \neq l, \\
   & ||\mathbf{b}_{m}-\mathbf{b}_l||_2 \geq d_{\mathrm{min}}, \forall m,l \in \mathcal{M} , m \neq l, \\
   & ||\mathbf{r}_{n}-\mathbf{r}_l ||_2 \geq d_{\mathrm{min}} , \forall n,l \in \mathcal{N} , n \neq l, \\
   &  0 \leq x_{u_0} \leq L-C, 0 \leq z_{u_0} \leq H-C,\\
   & 0 \leq x_{b_0} \leq L-C, 0 \leq z_{b_0} \leq H-C,
    \end{align}
\end{subequations}
where $\mathcal{M} = \{1,\dots,M\}$, $\mathcal{N} = \{1,\dots,N\}$, $\mathcal{K} = \{1,\dots,K\}$, $\mathbf{p} = [p_1,\dots,p_K]^T$, $\bm{\Omega} = [\bm{\omega}_1,\dots,\bm{\omega}_K]$, 
$B$ is the bandwidth of the channel, $R_{k,\mathrm{min}}$ is the minimum rate requirement of user $k$, $P_{k,\mathrm{max}}$ is the maximum transmit power of user $k$, $P_B$ is the circuit power consumption of the BS, $P_F$ is the
power consumption of FAR, and $d_{\mathrm{min}}$ is the minimum required distance between FAs to avoid mutual coupling\cite{ye2023fluid}. The minimum rate constraints are given by (\ref{sys1max0}a). (\ref{sys1max0}b) represents the minimum power constraints. (\ref{sys1max0}c)-(\ref{sys1max0}f) constrain that the FAs can only move within the corresponding regions, and (\ref{sys1max0}g)-(\ref{sys1max0}i) are the minimum distance constraints required to avoid mutual coupling between FAs. The feasible location region of the reference points $\mathbf{o}_u$ and $\mathbf{o}_b$ can be presented by constraints (\ref{sys1max0}j) and (\ref{sys1max0}k), respectively.

\section{Energy Efficiency Optimization\\ with Multiple Users}
In this section, we first decompose the problem into three sub-problems, i.e., the large-scale fading optimization sub-problem, the small-scale fading optimization sub-problem, and the joint optimization of power control and beamforming design sub-problem. To solve these problems, we propose Algorithm~\ref{Algorithm1}-Algorithm~\ref{Algorithm3}, and by iteratively solving them, problem \eqref{sys1max0} can be solved as shown in Algorithm~\ref{Algorithm4}.



\subsection{Large-scale Fading Optimization}
Large-scale fading, represented by $\beta_0$ and $\beta_k$, is mainly determined by the distance between the transmitter and the receiver. To find the optimal $\beta_0^{\mathrm{opt}}$ and $\beta_k^{\mathrm{opt}}$, we need to optimize the reference points of the FAR and keep the relative positions between FAs and the reference points fixed. Hence, the sub-problem can be formulated as 
\begin{subequations}\label{sub1}
    \begin{align} 
\mathop{\max}_{\mathbf{o}_u, \mathbf{o}_b \mathbf{U}, \mathbf{B}}\;&  
   B\sum_{k=1}^K\mathrm{log}_2(1+\gamma_k)
\tag{\ref{sub1}}\\
         \textrm{s.t.} \qquad 
   & \mathbf{u}_m^{\mathrm{cor}} - \mathbf{o}_u^{\mathrm{opt}} = \mathbf{u}_m - \mathbf{o}_u, \forall m \in \mathcal{M}, \\
   & \mathbf{b}_m^{\mathrm{cor}} - \mathbf{o}_b^{\mathrm{opt}} = \mathbf{b}_m - \mathbf{o}_b, \forall m \in \mathcal{M},\\
   \nonumber
   & (\ref{sys1max0}\mathrm{a}), (\ref{sys1max0}\mathrm{j}), (\ref{sys1max0}\mathrm{k}),
    \end{align}
\end{subequations}
where $\mathbf{o}_u^{\mathrm{opt}}$ and $\mathbf{o}_b^{\mathrm{opt}}$ are the optimal reference points of the FAR-U and the FAR-B, respectively, and $\mathbf{u}_m^{\mathrm{cor}}$ and $\mathbf{b}_m^{\mathrm{cor}}$ are the corresponding positions of the $m$-th FA at the FAR-U and the FAR-B, respectively. Since the positions of reference points $\mathbf{o}_u$ and $\mathbf{o}_b$ are uncoupled, we optimize them separately.

\subsubsection{Reference point optimization for $\mathbf{o}_b$} When the FA-position matrices $\mathbf{T}$ and $\mathbf{R}$, the transmit power vector $\mathbf{p}$, and the beamforming design matrix $\bm{\Omega}$ are given,
\eqref{SINR} can be written as
\begin{equation}\label{SINRNew}
    \gamma_k = \frac{\beta_kA_{kk}}{\beta_0^{-1}\sigma_r^2+ \sigma_u^2A_{k0}+\sum_{i=1,i\neq k}^K \beta_i A_{ki}},
\end{equation}
where $A_{kk} = p_k|\bm{\omega}_k^H\mathbf{Hh}_k|^2$, $A_{ki} = p_i|\bm{\omega}_k^H\mathbf{Hh}_i|^2$, and $A_{k0} =||\bm{\omega}_k^H\mathbf{H}||_2^2$ are constants. According to \eqref{SINRNew}, the optimal $\beta_0$ and corresponding reference point $\mathbf{o}_u^{\mathrm{opt}}$ can be given as 
\begin{equation}\label{Optimalob}
    \beta_0^{\mathrm{opt}} = \frac{\mu_0^2}{||\mathbf{o}_b^{{\mathrm{opt}}}-\mathbf{o}_r||_2^l}=  \frac{\mu_0^2}{(\sqrt{(L-C)^2+Y_0^2+(H-C)^2})^l},
\end{equation}
where $\mathbf{o}_b^{{\mathrm{opt}}} = [L-C,Y_0,H-C]$.

\subsubsection{Reference point optimization for $\mathbf{o}_u$} 
After obtaining the optimal solution to optimization for $\mathbf{o}_b$ and keeping the relative position relationships described in (\ref{sub1}a) and (\ref{sub1}b), we introduce the slack variable vector $\bm{\eta}=[\eta_1,\dots,\eta_K]^T$, and the large-scale fading sub-problem can be reformulated as
\begin{subequations}\label{sub11}
    \begin{align} 
\mathop{\max}_{\mathbf{o}_u,\mathbf{U},}\bm{\eta}\;&  
   B\sum_{k=1}^K\mathrm{log}_2(1+\eta_k)
\tag{\ref{sub11}}\\
         \textrm{s.t.} \qquad 
          &\eta_k\geq 2^{\frac{R_{\mathrm{k,min}}}{B}}-1, \forall k \in \mathcal{K}, \\
          & \eta_k \leq \frac{\beta_kA_{kk}}{\sum_{i=1,i\neq k}^K \beta_i A_{ki}+ \hat{\sigma}^2},\forall k \in \mathcal{K},\\
          \nonumber
   & (\ref{sys1max0}\mathrm{j}),(\ref{sub1}\mathrm{a}),
    \end{align}
\end{subequations}
where $\hat{\sigma}^2 = \beta_0^{-1}\sigma_r^2+ \sigma_u^2A_{k0}$.

The nonconvexity of problem \eqref{sub11} is caused by the constraint $(\ref{sub11}\mathrm{b})$. To handle the non-convex constraint $(\ref{sub11}\mathrm{b})$, we introduce slack variable $\zeta_k$, and constraint (\ref{sub11}b) is equivalent to 
\begin{equation}\label{sub11b1}
    \sum_{i=1,i\neq k}^K \frac{\mu_i}{||\mathbf{o}_u-\mathbf{t}_i||_2^l} A_{ki}   \leq \zeta_k- \hat{\sigma}^2, \forall k \in \mathcal{K},
\end{equation}
and
\begin{align}\label{sub11b2}
\nonumber
    \frac{\mu_k}{||\mathbf{o}_u-\mathbf{t}_k||_2^l}&A_{kk} \geq \eta_k \zeta_k\\
    &= \frac{1}{4}\left[(\eta_k+\zeta_k)^2-(\eta_k-\zeta_k)^2\right], \forall k \in \mathcal{K.}
\end{align}

Given the fact that function $\frac{1}{||\mathbf{o}_u-\mathbf{t}_k||_2^l}$ is either convex nor concave with respect to $\mathbf{o}_u$ with positive $l$, we adopt the SCA method to handle the nonconvexity of \eqref{sub11b1} and \eqref{sub11b2}. For \eqref{sub11b1}, it can be approximated as 
\begin{align}\label{sub11b3}
\nonumber
  & \sum_{i=1,i \neq k}^K \left[\frac{\mu_i A_{ki}}{||\mathbf{o}_u^{(n)}-\mathbf{t}_i||_2^l} -\frac{\mu_i l A_{ki}(\mathbf{o}_u^{(n)}-\mathbf{t}_i)^T(\mathbf{o}_u-\mathbf{o}_u^{(n)})}{||\mathbf{o}_u^{(n)}-\mathbf{t}_i||_2^{l+2}}\right] \\
   &+\hat{\sigma}^2\leq \zeta_k, \forall k \in \mathcal{K},
\end{align}
where the left hand side term is the first-order Taylor expansion of $\frac{\mu_i A_{ki}}{||\mathbf{o}_u-\mathbf{t}_i||_2^l}$ at point $\mathbf{o}_u^{(n)}$ and the subscript means the value of the variable in the $n$th iteration.

Similarly, to deal with the nonconvexity of \eqref{sub11b2}, we derive the first-order Taylor expansion of the terms on both hand sides, which can be given as
\begin{align}\label{sub11b4}
\nonumber
    &\frac{\mu_kA_{kk}}{||\mathbf{o}_u^{(n)}-\mathbf{t}_k||_2^l}-\frac{\mu_k l A_{kk}(\mathbf{o}_u^{(n)}-\mathbf{t}_k)^T(\mathbf{o}_u-\mathbf{o}_u^{(n)})}{||\mathbf{o}_u^{(n)}-\mathbf{t}_k||_2^{l+2}}\geq \\
    \nonumber
    & \frac{1}{4}\Big[( \eta_k + \zeta_k)^2 - ( \eta_k^{(n)} - \zeta_k^{(n)})^2+2(\eta_k^{(n)}-\zeta_k^{(n)})(\zeta_k-\zeta_k^{(n)}) \\
    &-2(\eta_k^{(n)} - \zeta_k^{(n)})(\eta_k-\eta_k^{(n)})\Big], \forall k \in \mathcal{K}.
\end{align}

With the above approximations, the non-convex problem in \eqref{sub11} can be reformulated as the approximated convex problem:
\begin{subequations}\label{sub12}
    \begin{align} 
\mathop{\max}_{\mathbf{o}_u,\mathbf{U},}\bm{\eta},\bm{\zeta}\;&  
   B\sum_{k=1}^K\mathrm{log}_2(1+\eta_k)
\tag{\ref{sub12}}\\
         \textrm{s.t.} \qquad 
         \nonumber
          &(\ref{sys1max0}\mathrm{j}),(\ref{sub1}\mathrm{a}),(\ref{sub11}\mathrm{a}),\eqref{sub11b3}, \eqref{sub11b4},
    \end{align}
\end{subequations}
where $\bm{\zeta} = [\zeta_1,\dots,\zeta_K]^T$. 

Thus, problem \eqref{sub11} can be tackled by utilizing SCA method, in which the approximated convex problem \eqref{sub12} is derived in each iteration.Moreover, for the original large-scale fading optimization sub-problem \eqref{sub1}, the details of solving it are presented in Algorithm~\ref{Algorithm1}.

\begin{algorithm}[t]
 	\caption{SCA method for Problem \eqref{sub1}}
 	  \label{Algorithm1}
    Initialize parameters $\mathbf{o}_u^{(0)}$, $\mathbf{o}_b^{(0)}$, $\mathbf{U}^{(0)}$, $\mathbf{B}^{(0)}$. \\
    Set iteration number $n = 1$.
    
    \BlankLine
    Obtain $\mathbf{o}_b^{\mathrm{opt}}$ by \eqref{Optimalob} and compute the $\mathbf{B}^{\mathrm{cor}}$ with obeying constraint (\ref{sub1}{b})\\ 
    \While{\eqref{sub12} does not converge}{
    Solve problem \eqref{sub12} with SCA method and obtain the output $\mathbf{o}_u^{(n)}$ {and $\mathbf{U}^{(n)}$} for this problem at this iteration.

   Set $n=n+1$.
    } 	
 \end{algorithm}

\subsection{Small-scale Fading Optimization}
When given reference points $\mathbf{o}_u$ and $\mathbf{o}_b$, transmit power vector $\mathbf{p}$, and beamforming design matrix $\bm{\Omega}$, problem \eqref{sys1max0} can be reformulated as

\begin{subequations}\label{sub2}
    \begin{align} 
        \mathop{\max}_{\mathbf{T}, \mathbf{U}, \mathbf{B}, \mathbf{R}} \quad &B\sum_{k=1}^K\mathrm{log}_2(1+\gamma_k),\tag{\ref{sub2}}\\
         \textrm{s.t.} \ 
        \nonumber
  &{(\ref{sys1max0}\mathrm{a}),}(\ref{sys1max0}\mathrm{c})-(\ref{sys1max0}\mathrm{i}),
    \end{align}
\end{subequations}
where $\gamma_k$ can be rewritten as 
\begin{equation}
    {\gamma}_k = \frac{|\bm{\omega}_k^H\mathbf{H}\mathbf{h}_k|^2}{\hat{\sigma}_{rk}^2+\hat{\sigma}_{uk}^2||\bm{\omega}_k^H{\mathbf{H}}||_2^2+\sum_{i=1,i\neq k}^K \hat{p}_{ik} |\bm{\omega}_k^H{\mathbf{Hh}}_i|^2},
\end{equation}
${\hat{\sigma}_{rk}^2} = \sigma_r^2(p_k\beta_0\beta_k)^{-1}$, ${\hat{\sigma}_{uk}^2} = \sigma_{u}^2(p_k\beta_k)^{-1}$, and $\hat{p}_{ik} = p_i\beta_i(p_k\beta_k)^{-1}$.

\subsubsection{Optimization of FA position of user $k$} With given $\mathbf{U}$, $\mathbf{B}$, $\mathbf{R}$, and $\{\mathbf{t}_{i},i \neq k\}_{i=1}^K$, problem \eqref{sub2} can be given as 
\begin{subequations}\label{sub21}
    \begin{align} 
    &\mathop{\max}_{\mathbf{t}_k, \bm{\tau}^{t_k}} \quad B\sum_{k=1}^K\mathrm{log}_2(1+\tau_k^{t_k}),\tag{\ref{sub21}}\\
    \textrm{s.t.} \ 
    &\tau_k^{t_k}\geq 2^{\frac{R_{k,\mathrm{min}}}{B}}-1,\forall k \in \mathcal{K} ,\\
    & \tau_k^{t_k} \leq \frac{|\hat{\bm{\omega}}_k^H\mathbf{h}_k|^2}{\hat{\sigma}_k^2+\sum_{i=1,i\neq k}^K \hat{p}_{ik} |\hat{\bm{\omega}}_k^H{\mathbf{h}}_i|^2}, \forall k \in \mathcal{K},\\
     & \mathbf{t}_k \in \mathcal{C}_{k},
    \end{align}
\end{subequations}
where $\bm{\tau}^{t_k} = [\tau_1^{t_k},\dots,\tau_K^{t_k}]^T$ is the slack variable vector, $\hat{\bm{\omega}}_k = \mathbf{H}^H \bm{\omega}_k$, and $\hat{\sigma}_k^2 = \hat{\sigma}_{rk}^2+\hat{\sigma}_{uk}^2||\bm{\omega}_k^H{\mathbf{H}}||_2^2$.

For any given $k$, constraint (\ref{sub21}b) is equivalent to
\begin{equation}\label{sub21b1}
\sqrt{(\hat{\sigma}_k^2+\sum_{i=1,i\neq k}^K \hat{p}_{ik} |\hat{\bm{\omega}}_k^H{\mathbf{h}}_i|^2) \tau_k^{t_k}} \leq \mathrm{Re}(\hat{\bm{\omega}}_k^H\mathbf{h}_k),
\end{equation}
whose left hand side term is concave with respect to $\tau_k^{t_k}$ and right hand side term is linear with respect to $\mathbf{h}_k$. Therefore, we expand the left hand side term of \eqref{sub21b1} with its first-order Taylor expansion as
\begin{align}\label{sub21b2}
\nonumber
    &\mathrm{Re}(\hat{\bm{\omega}}_k^H\mathbf{h}_k) \geq\\
    &{\sqrt{ A_k(\tau_k^{t_k})^{(n)}}+\frac{\sqrt{A_k}(\tau_k^{t_k}-(\tau_k^{t_k})^{(n)})}{2\sqrt{ (\tau_k^{t_k})^{(n)}}}}, \forall k \in \mathcal{K},
\end{align}
where $A_k = \hat{\sigma}_k^2+\sum_{i=1,i\neq k}^K \hat{p}_{ik} |\hat{\bm{\omega}}_k^H{\mathbf{h}}_i|^2$.

As a result, problem {\eqref{sub21}} can be reformulated as
\begin{subequations}\label{sub211}
    \begin{align} 
    \mathop{\max}_{\mathbf{t}_k, \bm{\tau}^{t_k}} \quad &B\sum_{k=1}^K\mathrm{log}_2(1+\tau_k^{t_k}),\tag{\ref{sub211}}\\
    \textrm{s.t.} \ 
    \nonumber
    &(\ref{sub21}\mathrm{a}),(\ref{sub21}\mathrm{c}),\eqref{sub21b2}{,}
    \end{align}
\end{subequations}
{which consists of convex objective function and convex constraints. Hence, convex problem \eqref{sub211} can be solved by existing optimization toolboxes, such as CVX.}

\subsubsection{Optimization of FA position of the FAR-U} When given $\mathbf{T}$, $\mathbf{B}$, $\mathbf{R}$, and $\{\mathbf{u}_{i}, i \neq m\}_{i=1}^M$, problem \eqref{sub2} can be given as 
\begin{subequations}\label{sub22}
    \begin{align} 
    &\mathop{\max}_{\mathbf{u}_m, \bm{\tau}^{u_m},\bm{\upsilon}^{u_m}} \quad B\sum_{k=1}^K\mathrm{log}_2(1+\tau_k^{u_m}),\tag{\ref{sub22}}\\
    \textrm{s.t.} \ 
    &\tau_k^{u_m}\geq 2^{\frac{R_{k,\mathrm{min}}}{B}}-1,\forall k \in \mathcal{K} ,\\
    & \upsilon_k^{u_m}\tau_k^{u_m} \leq |{\bm{\omega}}_k^H \mathbf{H}_0 \bm{\Theta}\mathbf{h}_k|^2, \forall k \in \mathcal{K},\\
    \nonumber
    &\hat{\sigma}_{rk}^2+\hat{\sigma}_{uk}^2||\bm{\omega}_k^H\mathbf{H}_0 \bm{\Theta}||_2^2+\sum_{i=1,i\neq k}^K \hat{p}_{ik} |{\bm{\omega}}_k^H \mathbf{H}_0 \bm{\Theta}{\mathbf{h}}_i|^2\\
    &\leq \upsilon_k^{u_m}, \forall k \in \mathcal{K},\\
   & \mathbf{u}_m \in \mathcal{C}_u,\\
   \nonumber
   &(\ref{sys1max0}\mathrm{g}),
    \end{align}
\end{subequations}
where $\tau_k^{u_m}$ and $\upsilon_k^{u_m}$ are slack variables.

{The nonconvexity of problem \eqref{sub22} lies in the non-convex constraints (\ref{sub22}b) and (\ref{sub22}c).} 
To handle the nonconvexity of (\ref{sub22}b), we first denote {that $\mathbf{h}_k^{u}(\mathbf{u}_m)=\bm{\Theta\mathbf{h}_k}$, and the right hand side term of constraint (\ref{sub22}b) can be written as}
{
\begin{equation}\label{sub22b0}
    |\bm{\omega}_k^H \mathbf{H}_0 \bm{\Theta}\mathbf{h}_k|^2
    =\mathbf{h}_k^{u}(\mathbf{u}_m)^H\hat{\mathbf{H}}_0\mathbf{h}_k^{u}(\mathbf{u}_m),
\end{equation}
}
where $\hat{\mathbf{H}}_0=\mathbf{H}_0^H\bm{\omega}_k\bm{\omega}_k^H \mathbf{H}_0$ is a constant {Hermitian} matrix.

{Although the right hand side term of \eqref{sub22b0} is neither convex nor concave with respect to $\mathbf{u}_m$, it is convex with respect to $\mathbf{h}_k^{u}(\mathbf{u}_m)$. Hence, we can obtain the lower bound of \eqref{sub22b0} as}
{
\begin{align}\label{sub22b1}
     \nonumber
    &\mathbf{h}_k^{u}(\mathbf{u}_m)^H\hat{\mathbf{H}}_0\mathbf{h}_k^{u}(\mathbf{u}_m)
    \geq  \mathbf{h}_k^{u}(\mathbf{u}_m^{(n)})^H\hat{\mathbf{H}}_0\mathbf{h}_k^{u}(\mathbf{u}_m^{(n)})\\
    \nonumber
    &+ 2\mathrm{Re}\left\{\mathbf{h}_k^{u}(\mathbf{u}_m^{(n)})^H\hat{\mathbf{H}}_0\left[\mathbf{h}_k^{u}(\mathbf{u}_m)-\mathbf{h}_k^{u}(\mathbf{u}_m^{(n)})\right]\right\}\\
    &=2\underbrace{\mathrm{Re}\left[\mathbf{h}_k^{u}(\mathbf{u}_m^{(n)})^H\hat{\mathbf{H}}_0\mathbf{h}_k^{u}(\mathbf{u}_m)\right]}_{\triangleq\tilde{h}_k^u(\mathbf{u}_m)}-\underbrace{\mathbf{h}_k^{u}(\mathbf{u}_m^{(n)})^H\hat{\mathbf{H}}_0\mathbf{h}_k^{u}(\mathbf{u}_m^{(n)})}_{\mathrm{constant}}.
\end{align}
where $\tilde{h}_k^u(\mathbf{u}_m)$ is linear with respect to $\mathbf{h}_k^u(\mathbf{u}_m)$, but is still not convex or concave to $\mathbf{u}_m$.}

{
\begin{lemma}
    Define $\nabla\tilde{h}_k^u(\mathbf{u}_m)$ as the gradient of $\tilde{h}_k^u(\mathbf{u}_m)$, $\nabla^2\tilde{h}_k^u(\mathbf{u}_m)$ is the Hessian matrix of $\tilde{h}_k^u(\mathbf{u}_m)$, and $\iota_k^u$ is a positive real number satisfying $\iota_k^u\mathbf{I}_3-\nabla^2\tilde{h}_k^u(\mathbf{u}_m)\succeq0$, whose derivations are provided in Appendix~\ref{Proof2}. Then, a global lower bound of $\tilde{h}_k^u(\mathbf{u}_m)$ is $\tilde{\tilde{h}}_k^u(\mathbf{u}_m)$ is obtained as
    \begin{align}\label{sub22b2}
    \nonumber
     &\tilde{\tilde{h}}_k^u(\mathbf{u}_m) =   
     -\frac{1}{2}\iota_k^u \mathbf{u}_m^T\mathbf{u}_m+\left(\nabla\tilde{h}_k^u(\mathbf{u}_m^{(n)})+\iota_k^u\mathbf{u}_m^{(n)}\right)^T\mathbf{u}_m\\
     &-\left(\nabla\tilde{h}_k^u(\mathbf{u}_m^{(n)})+\frac{1}{2}\iota_k^u\mathbf{u}_m^{(n)}\right)^T\mathbf{u}_m^{(n)}+\tilde{h}_k^u(\mathbf{u}_m^{(n)}).
    \end{align}
\end{lemma}
}
{
\begin{proof}
    See Appendix~\ref{proof3}.
\end{proof}
}

Then{,} constraint (\ref{sub22}b) can be approximated as
\begin{align}\label{sub22b3}
    \nonumber
    &{2\tilde{\tilde{h}}_k^u(\mathbf{u}_m)-\mathbf{h}_k^{u}(\mathbf{u}_m^{(n)})^H\hat{\mathbf{H}}_0\mathbf{h}_k^{u}(\mathbf{u}_m^{(n)})}\geq\\
    \nonumber
    &\frac{1}{4}\Big[( \upsilon_k^{u_m} + \tau_k^{u_m})^2 - ( (\upsilon_k^{u_m})^{(n)} - (\tau_k^{u_m})^{(n)})^2\\
    \nonumber
   &+ 2((\upsilon_k^{u_m})^{(n)}- (\tau_k^{u_m})^{(n)})(\tau_k^{u_m}-(\tau_k^{u_m})^{(n)})-\\
    &2((\upsilon_k^{u_m})^{(n)} - (\tau_k^{u_m})^{(n)})(\upsilon_k^{u_m}-(\upsilon_k^{u_m})^{(n)})\Big], \forall k \in \mathcal{K}{.}
\end{align}

{The nonconvexity of  (\ref{sub22}c) is caused by the non-convex terms $||\bm{\omega}_k^H\mathbf{H}_0 \bm{\Theta}||_2^2$ and $|\bm{\omega}_k^H\mathbf{H}_0 \bm{\Theta}\mathbf{h}_i|^2$. Mathematically, $||\bm{\omega}_k^H\mathbf{H}_0 \bm{\Theta}||_2^2$ can be viewed as a special case when the constant part $\mathbf{h}_i$ in $|\bm{\omega}_k^H\mathbf{H}_0 \bm{\Theta}\mathbf{h}_i|^2$ satisfies $\mathbf{h}_i\mathbf{h}_i^H=\mathbf{I}_M$. Hence, we first propose an upper bound, convex with respect to $\mathbf{u}_m$, of $|\bm{\omega}_k^H\mathbf{H}_0 \bm{\Theta}\mathbf{h}_i|^2\triangleq h_i^u(\mathbf{u}_m)$ in Lemma~\ref{lemma2} to handle its nonconvexity, and the convex upper bound of $||\bm{\omega}_k^H\mathbf{H}_0 \bm{\Theta}||_2^2$ can also be easily obtained.
\begin{lemma}\label{lemma2}
    $h_i^u(\mathbf{u}_m)$ can be globally upper bounded by $\tilde{\tilde{h}}_i^u(\mathbf{u}_m)$ defined in \eqref{proof42} in Appendix~\ref{proof4}, which is convex with respect to $\mathbf{u}_m$.
\end{lemma}
\begin{proof}
    See Appendix~\ref{proof4}.
\end{proof}

Define $\tilde{\tilde{h}}_0^u(\mathbf{u}_m)$ as the special case where $\mathbf{h}_i\mathbf{h}_i^H$ in $\tilde{\tilde{h}}_i^u(\mathbf{u}_m)$ equals $\mathbf{I}_M$, thus $\tilde{\tilde{h}}_0^u(\mathbf{u}_m)$ is an upper bound of $||\bm{\omega}_k^H\mathbf{H}_0 \bm{\Theta}||_2^2$.
}

Hence, constraint (\ref{sub22}c) can be reformulated as 
\begin{equation}\label{sub22b4}
\hat{\sigma}_{rk}^2+\hat{\sigma}_{uk}^2{\tilde{\tilde{h}}_0^u(\mathbf{u}_m)}+\sum_{i=1,i\neq k}^K \hat{p}_{ik}{\tilde{\tilde{h}}_i^u(\mathbf{u}_m)}\leq \upsilon_k^{u_m},\forall k \in \mathcal{K}.
\end{equation}

For the non-convex constraint (\ref{sys1max0}g), it can be approximated by the first-order Taylor expansion at the given point $\mathbf{u}_{m}^{(n)}$ as
\begin{equation}\label{constraintG}
 \frac{(\mathbf{u}_m^{{(n)}} - \mathbf{u}_l)^T(\mathbf{u}_m -  \mathbf{u}_l)}{||\mathbf{u}_m^{{(n)}} - \mathbf{u}_l||_2}\geq d_{\mathrm{min}},\forall m,l \in \mathcal{M}, m \neq l.
\end{equation}

Therefore, problem \eqref{sub22} can be reformulated as
\begin{subequations}\label{sub221}
    \begin{align} 
    \mathop{\max}_{\mathbf{u}_m, \bm{\tau}^{u_m},\bm{\upsilon}^{b_m}} \quad &B\sum_{k=1}^K\mathrm{log}_2(1+\tau_k^{u_m}),\tag{\ref{sub221}}\\
    \textrm{s.t.} \ 
    \nonumber
   &(\ref{sub22}\mathrm{a}),(\ref{sub22}\mathrm{d}),{\eqref{sub22b1}},{\eqref{sub22b4}},\eqref{constraintG},
    \end{align}
\end{subequations}
{which is convex and can be solved by existing toolboxes.
}

\subsubsection{Optimization of FA position of the FAR-B} With given $\mathbf{T}$, $\mathbf{U}$, $\mathbf{R}$, and $\{\mathbf{b}_{i}, i \neq m\}_{i=1}^M$, problem \eqref{sub2} can be rewritten as 
\begin{subequations}\label{sub23}
    \begin{align} 
    &\mathop{\max}_{\mathbf{b}_m, \bm{\tau}^{b_m},\bm{\upsilon}^{b_m}} \quad B\sum_{k=1}^K\mathrm{log}_2(1+\tau_k^{b_m}),\tag{\ref{sub23}}\\
    \textrm{s.t.} \ 
    &\tau_k^{b_m}\geq 2^{\frac{R_{k,\mathrm{min}}}{B}}-1,\forall k \in \mathcal{K} ,\\
    & \upsilon_k^{b_m}\tau_k^{b_m} \leq |{\bm{\omega}}_k^H \mathbf{H}_0 \bm{\Theta}\mathbf{h}_k|^2, \forall k \in \mathcal{K},\\
    \nonumber
    &\hat{\sigma}_{rk}^2+\hat{\sigma}_{uk}^2||\bm{\omega}_k^H{\mathbf{H}_0} \bm{\Theta}||_2^2\\
    &+\sum_{i=1,i\neq k}^K \hat{p}_{ik} |{\bm{\omega}}_k^H \mathbf{H}_0 \bm{\Theta}{\mathbf{h}}_i|^2\leq \upsilon_k^{b_m}, \forall k \in \mathcal{K},\\
    & \mathbf{b}_m \in \mathcal{C}_b,\\
    \nonumber
   & (\ref{sys1max0}\mathrm{h}),
    \end{align}
\end{subequations}
where $\tau_k^{b_m}$ and $\upsilon_k^{b_m}$ are slack variables.

{
Problem \eqref{sub23} is non-convex, because the constraints (\ref{sub23}b) and (\ref{sub23}c) are not convex. We first deal with the nonconvexity of (\ref{sub23}b). Define $\mathbf{h}_k^b(\mathbf{b}_m)=\mathbf{H}_0\bm{\Theta\mathbf{h}_k}$, where $\mathbf{h}_k$ is a constant vector, we can rewrite the right hand side term of constraint (\ref{sub23}b) as
\begin{equation}\label{sub23b0}
    |{\bm{\omega}}_k^H \mathbf{H}_0 \bm{\Theta}\mathbf{h}_k|^2=\mathbf{h}_k^b(\mathbf{b}_m)^H\hat{\bm{\Omega}}_k\mathbf{h}_k^b(\mathbf{b}_m),
\end{equation}
where $\hat{\bm{\Omega}}_k=\bm{\omega}_k\bm{\omega}_k^H$. 

The right hand side term of \eqref{sub23b0} holds a similar form of \eqref{sub22b0}. Hence, we can find a global lower bound $\tilde{\tilde{h}}_k^b(\mathbf{b}_m)$ of it utilizing the same method used to find $\tilde{\tilde{h}}_k^u(\mathbf{u}_m)$.
}

Then constraint (\ref{sub23}b) can be {reformulated} as 
\begin{align}\label{sub23b1}
\nonumber
    &{\tilde{\tilde{h}}_k^b(\mathbf{b}_m)} \geq\frac{1}{4}\Big[( \upsilon_k^{b_m} + \tau_k^{b_m})^2 - ( (\upsilon_k^{b_m})^{(n)} - (\tau_k^{b_m})^{(n)})^2\\
    \nonumber
    &+ 2((\upsilon_k^{b_m})^{(n)}- (\tau_k^{b_m})^{(n)})(\tau_k^{b_m}-(\tau_k^{b_m})^{(n)})\\
    &-2((\upsilon_k^{b_m})^{(n)} - (\tau_k^{b_m})^{(n)})(\upsilon_k^{b_m}-(\upsilon_k^{b_m})^{(n)})\Big],\forall k \in 
    \mathcal{K}.
\end{align}

For constraint (\ref{sub23}c), the {nonconvexity of it lies in the terms $||{\bm{\omega}}_k^H \mathbf{H}_0 \bm{\Theta}||_2^2\triangleq h_0^b(\mathbf{b}_m)$ and $|{\bm{\omega}}_k^H \mathbf{H}_0 \bm{\Theta}{\mathbf{h}}_i|^2\triangleq h_i^b(\mathbf{b}_m)$. Similar to the derivations in \eqref{proof61} and \eqref{proof62} in Appendix~\ref{proof6}, we can find $\tilde{\tilde{h}}_i^u(\mathbf{b}_m)$ to upper bound $||{\bm{\omega}}_k^H \mathbf{H}_0 \bm{\Theta}||_2^2$, where $\tilde{\tilde{h}}_i^b(\mathbf{b}_m)$ is given as
\begin{align}\label{sub23b2}
   \nonumber
   \tilde{\tilde{h}}_i^b(\mathbf{b}_m) &=  \frac{1}{2}\iota_i^b\mathbf{b}_m^T\mathbf{b}_m+\left(\nabla h_i^b(\mathbf{b}_m^{(n)})-\iota_i^b\mathbf{b}_m^{(n)}\right)^T\mathbf{b}_m\\
    &+\left(\frac{1}{2}\iota_i^b\mathbf{b}_m^{(n)}-\nabla{h}_i^b(\mathbf{b}_m^{(n)})\right)^T\mathbf{b}_m^{(n)}+h_i^b(\mathbf{b}_m^{(n)}),
\end{align}
where $\iota_i^b\geq ||\nabla^2 h_i^b(\mathbf{b}_m)||_2$.

Moreover, by replacing $\mathbf{h}_i\mathbf{h}_i^H$ in $\tilde{\tilde{h}}_i^b(\mathbf{b}_m)$ with $\mathbf{I}_M$, we can obtain the upper bound of $||{\bm{\omega}}_k^H \mathbf{H}_0 \bm{\Theta}||_2^2$ as $\tilde{\tilde{h}}_0^b(\mathbf{b}_m)$.
}

As a result, constraint (\ref{sub23}c) can be {approximated} as 
\begin{equation}\label{sub23b3}
\hat{\sigma}_{rk}^2+\hat{\sigma}_{uk}^2{\tilde{\tilde{h}}_0^b(\mathbf{b}_m)}+\sum_{i=1,i\neq k}^K \hat{p}_{ik}{\tilde{\tilde{h}}_i^b(\mathbf{b}_m)}\leq \upsilon_k^{b_m},\forall k \in \mathcal{K}.
\end{equation}

For the non-convex constraint (\ref{sys1max0}h), it can be approximated by the first-order Taylor expansion at the given point $\mathbf{b}_{m}^{(n)}$ as
\begin{equation}\label{constraintH}
   \frac{(\mathbf{b}_m^{{(n)}} - \mathbf{b}_l)^T(\mathbf{b}_m -  \mathbf{b}_l)}{||\mathbf{b}_m^{{(n)}} - \mathbf{b}_l||_2}\geq d_\mathrm{min}, \forall m , l \in \mathcal{M}, l \neq m.
\end{equation}

Finally, problem \eqref{sub23} can be {approximated by the convex form} as
\begin{subequations}\label{sub231}
    \begin{align} 
    \mathop{\max}_{\mathbf{b}_m, \bm{\tau}^{b_m},\bm{\upsilon}^{b_m}} \quad &B\sum_{k=1}^K\mathrm{log}_2(1+\tau_k^{b_m}),\tag{\ref{sub231}}\\
    \textrm{s.t.} \ 
    \nonumber
    &(\ref{sub23}\mathrm{a}),(\ref{sub23}\mathrm{d}),{\eqref{sub23b1},\eqref{sub23b3}},\eqref{constraintH},
    \end{align}
\end{subequations}
{which can be tackled by the existing convex toolboxes.}

\subsubsection{Optimization of FA position of BS} With given $\mathbf{T}$, $\mathbf{U}$, $\mathbf{B}$, and $\{\mathbf{r}_{i}, i \neq n\}_{i=1}^N$, problem \eqref{sub2} can be reformulated as 
\begin{subequations}\label{sub24}
    \begin{align} 
    &\mathop{\max}_{\mathbf{r}_n, \bm{\tau}^{r_n},\bm{\upsilon}^{r_n}} \quad B\sum_{k=1}^K\mathrm{log}_2(1+\tau_k^{r_n}),\tag{\ref{sub24}}\\
    \textrm{s.t.} \ 
    &\tau_k^{r_n}\geq 2^{\frac{R_{k,\mathrm{min}}}{B}}-1,\forall k \in \mathcal{K} ,\\
    & \upsilon_k^{r_n}\tau_k^{r_n} \leq  {h_k^{r_n}}, \forall k \in \mathcal{K},\\
    &\hat{\sigma}_{rk}^2+\hat{\sigma}_{uk}^2 {h_0^{r_n}}+\sum_{i=1,i\neq k}^K \hat{p}_{ik} { h_i^{r_n}}\leq \upsilon_k^{r_n}, \forall k \in \mathcal{K},\\
     &  \mathbf{r}_{n}\in \mathcal{C}_r,\\
     \nonumber
   &(\ref{sys1max0}\mathrm{i}),
    \end{align}
\end{subequations}
where {$h_l^{r_n}=|{\bm{\omega}}_k^H \mathbf{H}_0 \bm{\Theta}\mathbf{h}_l|^2,l=\{k,i\}$, $h_0^{r_n}=||{\bm{\omega}}_k^H \mathbf{H}_0 \bm{\Theta}||_2^2$}, and $\tau_k^{r_n}$ and $\upsilon_k^{r_n}$ are slack variables.

{
For the non-convex constraints (\ref{sub24}b) and (\ref{sub24}c), we can acquire a lower bound $\tilde{\tilde{h}}_k^{r_n}$ of the right hand side of constraint (\ref{sub24}b) adopting a similar way shown in \eqref{sub22b1} and \eqref{sub22b2}, and obtain $\tilde{\tilde{h}}_i^{r_n}$ and $\tilde{\tilde{h}}_0^{r_n}$ to the upper bound $h_i^{r_n}$ and $h_0^{r_n}$ in the constraint (\ref{sub24}b) with the same method used in \eqref{sub23b2}.
}

Therefore, {constraints} (\ref{sub24}b) {and  (\ref{sub24}c)} can be {respectively} approximated as 
\begin{align}\label{sub24b1}
\nonumber
    &{\tilde{\tilde{h}}_k^{r_n}} \geq\frac{1}{4}\Big[( \upsilon_k^{r_n} + \tau_k^{r_n})^2 - ( (\upsilon_k^{r_n})^{(n)} - (\tau_k^{r_n})^{(n)})^2\\
    \nonumber
    &+ 2((\upsilon_k^{r_n})^{(n)}- (\tau_k^{r_n})^{(n)})(\tau_k^{r_n}-(\tau_k^{r_n})^{(n)})\\
    &-2((\upsilon_k^{r_n})^{(n)} - (\tau_k^{r_n})^{(n)})(\upsilon_k^{r_n}-(\upsilon_k^{r_n})^{(n)})\Big], \forall k \in \mathcal{K},
\end{align}
{and}
\begin{align}\label{sub24b2}
\hat{\sigma}_{rk}^2+\hat{\sigma}_{uk}^2 {\tilde{\tilde{h}}_0^{r_n}}+\sum_{i=1,i\neq k}^K \hat{p}_{ik} { \tilde{\tilde{h}}_i^{r_n}}\leq \upsilon_k^{r_n}, \forall k \in \mathcal{K}{.}
\end{align}

Besides, non-convex constraint (\ref{sys1max0}i) can be approximated by its first-order Taylor expansion at the given point $\mathbf{r}_{n}^{(n)}$ as
\begin{equation}\label{constraintI}
  \frac{(\mathbf{r}_n^{(n)} - \mathbf{r}_l)^T(\mathbf{r}_n -  \mathbf{r}_l)}{||\mathbf{r}_n^{(n)} - \mathbf{r}_l||_2}\geq d_0, \forall n,l \in \mathcal{N}, n\neq l.
\end{equation}

{Accordingly}, problem \eqref{sub24} can be approximated as 
\begin{subequations}\label{sub241}
    \begin{align} 
    \mathop{\max}_{\mathbf{r}_n, \bm{\tau}^{r_n},\bm{\upsilon}^{r_n}} \quad &B\sum_{k=1}^K\mathrm{log}_2(1+\tau_k^{r_n}),\tag{\ref{sub241}}\\
    \textrm{s.t.} \ 
    \nonumber
    &(\ref{sub24}\mathrm{a}),(\ref{sub24}\mathrm{d}),{\eqref{sub24b1}-}\eqref{constraintI}.
    \end{align}
\end{subequations}

{Since problems \eqref{sub211}, \eqref{sub221}, \eqref{sub231}, and \eqref{sub241} are all convex, they can be solved with existing tool box, respectively.} By alternately {solving} the four sub-problems, we can obtain the {sub-optimal} solution {to} the small-scale fading optimization with the proposed alternating optimization (AO) algorithm shown in Algorithm~\ref{Algorithm2}.

 \begin{algorithm}[t]
 	\caption{AO for Problem \eqref{sub2}}
 	  \label{Algorithm2}
    Initialize parameters $\mathbf{T}^{(0)}$, $\mathbf{R}^{(0)}$, $\mathbf{U}^{(0)}$, $\mathbf{B}^{(0)}$. \\
    Set iteration number $n = 1$, and $i=1$.
    \BlankLine
    \While{\eqref{sub2} does not converge}{
    \While{$i \neq K$}{
    Solve problem \eqref{sub211} with given $\mathbf{U}$, $\mathbf{B}$, $\mathbf{R}$, and $\{\mathbf{t}_{k},k\neq i\}_{k=1}^K$ to update $\mathbf{t}_i$. Set $i=i+1$. 
    }
     \While{$i \neq M$}{
    
    Solve problem \eqref{sub221} with given $\mathbf{T}$, $\mathbf{B}$, $\mathbf{R}$, and $\{\mathbf{u}_{m},m\neq i\}_{m=1}^M$ to update $\mathbf{u}_i$. Set $i = i + 1$. 
    }
    
     \While{$i \neq M$}{
    
    Solve problem \eqref{sub231} with given $\mathbf{T}$, $\mathbf{U}$, $\mathbf{R}$, and $\{\mathbf{b}_{m},m\neq i\}_{m=1}^M$ to update $\mathbf{b}_i$. Set $i = i + 1$.
    }  
     \While{$i \neq N$}{
    Solve problem \eqref{sub241} with given $\mathbf{T}$, $\mathbf{U}$, $\mathbf{B}$, and $\{\mathbf{r}_{n},n\neq i\}_{n=1}^N$ to update $\mathbf{r}_i$. Set $i = i + 1$.
    }
    Set iteration number $n=n+1$. 
    
    Reset $i=1$.
    }
    
 \end{algorithm}
 
\subsection{Joint Optimization of Power Control and Beamforming
Design}
When the positions of FAs are fixed, problem \eqref{sub3}
can be formulated as
\begin{subequations}\label{sub3}
    \begin{align} 
       & \mathop{\max}_{\mathbf{p}, \bm{\Omega},\bm{\psi} }\quad \dfrac{B\sum_{k=1}^K\mathrm{log}_2(1+\psi_k)}{\sum_{k=1}^Kp_k + P_B + P_F} ,\tag{\ref{sub3}}\\
         \textrm{s.t.} \  
         \nonumber
         & \psi_k \leq \frac{p_k|\bm{\omega}_k^H\bar{\mathbf{h}}_k|^2}{\bar{\sigma}_{rk}^2+\bar{\sigma}_{uk}^2||\bm{\omega}_k^H\mathbf{H}||_2^2+\sum_{i=1,i\neq k}^K p_i\bar{\beta}_{ik} |\bm{\omega}_k^H\bar{\mathbf{h}}_i|^2},\\
         &\forall k \in \mathcal{K},\\
         &\psi_k \geq 2^{\frac{R_{k,\mathrm{min}}}{B}-1},\forall k \in \mathcal{K}, \\
   & 0\leq p_k\leq P_{k,\mathrm{max}},\forall k \in \mathcal{K},
    \end{align}
\end{subequations}
where $\bm{\psi} = [\psi_1,\dots,\psi_K]^T$ is the slack vector, and $\bar{\mathbf{h}}_k = \mathbf{Hh}_k$, $\bar{\sigma}_{rk}^2 = \sigma_{r}^2(\beta_0\beta_k)^{-1}$, $\bar{\sigma}_{{uk}}^2 = \sigma_{u}^2 \beta_k^{-1}$ and $\bar{\beta}_{ik} = \beta_i \beta_k^{-1}$ are all constants. 

To tackle the nonconvexity of (\ref{sub3}a), we introduce a slack variable $\chi_k>0$ and reformulate it as
\begin{equation}\label{sub3a1}
p_k|\bm{\omega}_k^H\bar{\mathbf{h}}_k|^2\geq \psi_k\chi_k,
\end{equation}
and
\begin{equation}\label{sub3a2}
 \bar{\sigma}_{rk}^2+\bar{\sigma}_{uk}^2||\bm{\omega}_k^H\mathbf{H}||_2^2+\sum_{i=1,i\neq k}^K p_i \bar{\beta}_{ik}|\bm{\omega}_k^H\bar{\mathbf{h}}_i|^2 \leq \chi_k.  
\end{equation}

Without loss of generality, the term $\bm{\omega}_k^H\bar{\mathbf{h}}_k$ in constraint
\eqref{sub3a1} can be expressed as a real number through an arbitrary rotation to beamforming $\bm{\omega}_k${\cite{9461768,10032275}.} As a result, constraint \eqref{sub3a1} can be equivalently expressed as
{
\begin{equation}\label{sub3a3}
   \bm{\omega}_k^H \bar{\mathbf{h}}_k \geq \dfrac{\sqrt{\psi_k \chi_k}}{\sqrt{p_k}}.
\end{equation}
}
Since $\bm{\omega}_k^H {\bar{\mathbf{h}_k}}$ is linear, we utilize first-order Taylor series to expand its non-convex term. Then, \eqref{sub3a3} can be rewritten as

{
\begin{align}\label{sub3a4}
\nonumber
     \bm{\omega}_k^H \bar{\mathbf{h}}_k &\geq \dfrac{\sqrt{\chi_k^{(n)}}(\psi_k-\psi_k^{(n)})}{2\sqrt{p_k^{(n)}\psi_k^{(n)}}}+\dfrac{\sqrt{\psi_k^{(n)}}(\chi_k-\chi_k^{(n)})}{2\sqrt{p_k^{(n)}\chi_k^{(n)}}}\\
    &- \dfrac{\sqrt{\psi_k^{(n)}\chi_k^{(n)}}(p_k-p_k^{(n)})}{2(p_k^{(n)})^{\frac{3}{2}}}+\frac{\sqrt{\psi_k^{(n)}\chi_k^{(n)}}}{\sqrt{p_k^{(n)}}}.
\end{align}
}

{For constraint \eqref{sub3a2}, we first prove that the term $\bar{\sigma}_u^2||\bm{\omega}_k^H\mathbf{H}||_2^2$ is convex with respect to $\bm{\omega}_k$.

The Hessian matrix of $||\bm{\omega}_k^H\mathbf{H}||_2^2$ can be derived as
\begin{equation}
    \nabla^2 \bm{\omega}_k^H\mathbf{H}\mathbf{H}^H\bm{\omega}_k=\mathbf{H}\mathbf{H}^H.
\end{equation}
For any non-zero vector $\mathbf{v}$, $\mathbf{v}^H\mathbf{H}\mathbf{H}^H\mathbf{v}=||\mathbf{H}^H\mathbf{v}||_2^2\geq0$. Hence, $\bar{\sigma}_{uk}^2||\bm{\omega}_k^H\mathbf{H}||_2^2$ is convex over $\bm{\omega}_k$. Similarly, we can prove that $|\bm{\omega}_k^H\bar{\mathbf{h}}_i|^2$ is also convex with respect to $\bm{\omega}_k$.
}

Then, the nonconvexity of \eqref{sub3a2} is caused by the term $\sum_{i=1,i\neq k}^K p_i \bar{\beta}_{ik}|\bm{\omega}_k^H\bar{\mathbf{h}}_i|^2$. To solve this, we first rewrite it in the following equivalent form.
\begin{equation}\label{sub3a5}
    \frac{1}{4}\sum_{i=1,i\neq k}^K \left[\left(p_i\bar{\beta}_{ik}+|\bm{\omega}_k^H\bar{\mathbf{h}}_i|^2\right){^2}-\left(p_i\bar{\beta}_{ik}-|\bm{\omega}_k^H\bar{\mathbf{h}}_i|^2\right){^2}\right].
\end{equation}

Given the fact that {$\left(p_i\bar{\beta}_{ik}+|\bm{\omega}_k^H\bar{\mathbf{h}}_i|^2\right){^2}$ is jointly convex with respect to $p_i$ and $\bm{\omega}_k$, we replace the non-concave term $\left(p_i\bar{\beta}_{ik}-|\bm{\omega}_k^H\bar{\mathbf{h}}_i|^2\right){^2}$ with its first-order Taylor approximation.} Then, \eqref{sub3a5} can be written as
\begin{align} \label{sub3a6}
\nonumber
    &\frac{1}{4}\sum_{i=1,i\neq k}^K (p_i\bar{\beta}_{ik}+|\bm{\omega}_k^H\bar{\mathbf{h}}_i|{^2})^2-\Big[p_i^{(n)}\bar{\beta}_{ik}-|(\bm{\omega}_k^H)^{(n)}\bar{\mathbf{h}}_i|^2\Big]^2\\
    \nonumber
    &{-2}\bar{\beta}_{ik}\left[p_i^{(n)}\bar{\beta}_{ik} - |(\bm{\omega}_k^H)^{(n)}\bar{\mathbf{h}}_i|^2\right]  \left(p_i-p_i^{(n)}\right)  \\
     \nonumber
    & {+}4\left[p_i^{(n)}\bar{\beta}_{ik} \right. - |(\bm{\omega}_k^H)^{(n)}\bar{\mathbf{h}}_i|^2\Big]{\mathrm{Re}\left[\bar{\mathbf{h}}_i^H\bm{\omega}_k^{(n)}\bar{\mathbf{h}}_i^T(\bm{\omega}_k-\bm{\omega}_k^{(n)})\right]} \\
    & +\bar{\sigma}_{uk}^2||\bm{\omega}_k^H\mathbf{H}||_2^2+\bar{\sigma}_{rk}^2\leq \chi_k, \forall k \in \mathcal{K}.
\end{align}

As a result, problem
\eqref{sub3} can be reformulated as
\begin{subequations}\label{sub31}
    \begin{align} 
       & \mathop{\max}_{\mathbf{p}, \bm{\Omega},\bm{\psi},\bm{\chi}}\quad \dfrac{B\sum_{k=1}^K\mathrm{log}_2(1+\psi_k)}{\sum_{k=1}^Kp_k + P_B + P_F} ,\tag{\ref{sub31}}\\
         \textrm{s.t.} \  
         \nonumber
         &(\ref{sub3}\mathrm{b}), (\ref{sub3}\mathrm{c}),\eqref{sub3a4}, \eqref{sub3a6}{.}
    \end{align}
\end{subequations}

The objective function of problem \eqref{sub31} is {a fraction composed of} a concave function divided by a convex function, and the feasible set is convex. Therefore, the optimal solution of problem \eqref{sub31} can be obtained using the Dinkelbach method\cite{dinkelbach1967nonlinear}. As a result, the joint optimization problem on power allocation and beamforming design can be solved using the SCA method, where Dinkelbach is utilized in each iteration. Specific details are shown in Algorithm \ref{Algorithm3}.

\begin{algorithm}[t]
 	\caption{Dinkelbach method based algorithm for problem \eqref{sub3} with SCA}
 	  \label{Algorithm3}
   Initialized parameters $\mathbf{p}^{(0)}$, $\bm{\Omega}^{(0)}$ {$\bm{\chi}^{(0)},\bm{\psi}^{(0)}$}. Set $n = 1$.
    \BlankLine 
    \While{\eqref{sub31} does not converge}{
    The objective function of problem \eqref{sub3} is equivalent to be
    $\underbrace{B\sum_{k=1}^K\mathrm{log}_2(1+\psi_k)}_{f(\mathcal{S})} - s^{(t)}(\underbrace{\sum_{k=1}^K p_k + P_B + P_F}_{g(\mathcal{S})})$, where $\mathcal{S} = \{\mathbf{p},\bm{\Omega},\bm{\chi},\bm{\psi}\}$, $s^{(t)}$ is determined by the iteration algorithm as shown in \cite{dinkelbach1967nonlinear}:\\
    Initialize $t={0}$,$s^{(0)}$, $\mathcal{S}^{(0)} = {\mathcal{S}^{(n-1)}} $ and $\varepsilon_3 > 0$\\
    \While{$|f(\mathcal{S}^{(t)}) - s^{(t)}\ g(\mathcal{S}^{(t)})| < \varepsilon_3$}{
     $t=t+1$.
     
    $s^{(t)} = \frac{f(\mathcal{S}^{(t-1)})}{g(\mathcal{S}^{(t-1)})}$.

    $\mathcal{S}^{(t)} = \mathrm{arg}\ \mathrm{max} \ f(\mathcal{S}^{}) - s^{(t)}\ g(\mathcal{S})$.

    }
    Solve problem \eqref{sub3} with approximated problem \eqref{sub31}, whose objective function is equivalently replaced by $f(\mathcal{S}^{(t)}) - s^{(t)}\ g(\mathcal{S}^{(t)})$, and obtain the output at this iteration. 
    
    Set $n=n+1.$
    }
 \end{algorithm}
\subsection{Complexity Analysis}
The overall algorithm to maximize EE of the proposed system is presented in Algorithm~\ref{Algorithm4}. For large-scale optimization sub-problem \eqref{sub1}, the main complexity lies in solving problem \eqref{sub11} with SCA method. Since there are $K{+2M}+2$ variables in \eqref{sub11}, the complexity of iterations required for SCA method is $\mathcal{O}(\sqrt{K{+2M}+2}\mathrm{log}_2(1/\epsilon_1))$\cite{9497709}, where $\epsilon_1$ is the accuracy of the SCA method for solving problem \eqref{sub11}. {In} each iteration, the complexity of solving problem \eqref{sub12} is $\mathcal{O}(S_1^2S_2)$\cite{lobo1998applications}, where $S_1=2K{+2M}+2$ is the total number of variables of problem \eqref{sub12}, and $S_2=3K{+M}+4$ is the total number of constraints of problem \eqref{sub12}. Hence, the complexity of solving problem \eqref{sub1} is $\mathcal{O}(K^{3.5}\mathrm{log}_2(1/\epsilon_1))$. For small-scale optimization problem, the total complexity is $\mathcal{O}(K^{4}\mathrm{log}_2(1/\epsilon_{21})+MK^{3}\mathrm{log}_2(1/\epsilon_{22})+MK^{3}\mathrm{log}_2(1/\epsilon_{23})+NK^3\mathrm{log}_2(1/\epsilon_{24}))\triangleq\mathcal{O}(Q)$, where $\epsilon_{21}$, $\epsilon_{22}$, $\epsilon_{23}$, and $\epsilon_{24}$ are the accuracies of the SCA method for solving problem \eqref{sub21}, \eqref{sub22}, \eqref{sub23}, and \eqref{sub24}, respectively. Similarly, the total complexity of solving problem \eqref{sub3} is $\mathcal{O}(S_3K^5\mathrm{log}_2(1/\epsilon_3))$, where $S_3$ is the number of iterations for solving problem \eqref{sub3} with the Dinkelbach method, and $\epsilon_3$ is the accuracy required of the SCA method. As a result, the total complexity of Algorithm~\ref{Algorithm4} for solving
problem \eqref{sys1max0} is $\mathcal{O}(S_4K^{3.5}\mathrm{log}_2(1/\epsilon_1)+S_4Q+S_4S_3K^5\mathrm{log}_2(1/\epsilon_3))$, where $S_4$ is the number of iterations of Algorithm~\ref{Algorithm4}.

 \begin{algorithm}[t]
 	\caption{Iterative Optimization for Problem \eqref{sys1max0}}
 	  \label{Algorithm4}
    
    \BlankLine
    Initialize parameters $\mathbf{T}^{(0)}$, $\mathbf{U}^{(0)}$, $\mathbf{B}^{(0)}$, 
    $\mathbf{R}^{(0)}$, $\mathbf{o}_u^{(0)}$, $\mathbf{o}_b^{(0)}$,
    $\mathbf{P}^{(0)}$, $\mathbf{\bm{\Omega}^{(0)}}$. Set iteration number $n=1$.
    
    \While{\textnormal{Objective value of \eqref{sys1max0} does not converge}}{
    Solve problem \eqref{sub1} by Algorithm~\ref{Algorithm1} with given $\mathbf{T}^{(n-1)}$, $\mathbf{R}^{(n-1)}$, $\mathbf{P}^{(n-1)}$ and $\bm{\Omega}^{(n-1)}$. Obtain the solution denoted as $ \mathbf{o}_u^{(n)}$, $ \mathbf{o}_b^{(n)}$, $\mathbf{U}^{(n)}$, and $\mathbf{B}^{(n)}$.
    
    Solve problem \eqref{sub2} by Algorithm~\ref{Algorithm2} with given $ \mathbf{o}_u^{(n)}$, $ \mathbf{o}_b^{(n)}$, $\mathbf{P}^{(n-1)}$ and $\bm{\Omega}^{(n-1)}$, and obtain the solution denoted as $ \mathbf{R}^{(n)}$, $ \mathbf{T}^{(n)}$, $\mathbf{U}^{(n)}$, and $\mathbf{B}^{(n)}$.

     Solve problem \eqref{sub3} by Algorithm~\ref{Algorithm3} with given $ \mathbf{o}_u^{(n)}$, $ \mathbf{o}_b^{(n)}$, $ \mathbf{T}^{(n)}$, $\mathbf{U}^{(n)}$, and $\mathbf{B}^{(n)}$, and obtain the solution denoted as $ \mathbf{P}^{(n)}$, and $ \mathbf{\Omega}^{(n)}$.

    Calculate the value of the objective function (\ref{sys1max0}). 
    
    Set $n=n+1$.
    }
 		
 \end{algorithm}

\section{Numerical Simulation Results}

In our simulations, we assume that all elevation angles and azimuth angles are i.i.d, randomly distributed in $[0,\pi]$, the maximum transmit power {and minimum required rate} of {all users} {are equal}, i.e., $P_{1,\mathrm{max}}=\cdots=P_{K,\mathrm{max}}{=P_{\mathrm{max}}}$ { and $R_{1,\mathrm{min}}=\cdots=R_{K,\mathrm{min}} = R_{\mathrm{min}}$}, and wave length $\lambda = 0.3\mathrm{m}$. We set the conductivity of the FAR is $\sigma=10^{-14}$, permittivity $\varepsilon = 4\varepsilon_v$, and permeability $\mu = \mu_v$, where $\varepsilon_v$ and $\mu_v$ are the vacuum permittivity and vacuum permeability, respectively. The location of BS is set as $(0, 0, 0)$, the distance between FAR and BS along y-axis $Y_0 = 50\mathrm{m}$, and the users are uniformly in a circle centered at $(d_x, d_y, 0)$ with radius $ d_0$, as shown in Fig.~\ref{SystemModel}. We compare our proposed algorithm using FAR (labeled as `FAR') with the scheme with STAR-RIS\cite{mu2021simultaneously} deployed on both sides of the blockage (labeled as `SRIS'), and the conventional AF scheme\cite{6671453} (labeled as `AFR') on both sides of the blockage. In particular, STAR-RIS is set as `Energy Splitting' mode\cite{mu2021simultaneously}, and the total numbers of reflection and transmission elements are $2M$ with { equal reflection coefficients $\varrho^r$ and transmission coefficients $\varrho^t$, i.e., $\varrho^t_m = {0.9}, \varrho^r_m = {0.1},\forall m \in \mathcal{M}$.} The circuit power of each STAR-RIS element is denoted as $P_S$. For AFR, we consider that there are $M$ antennas on both sides of the blockage, with the circuit power of each antenna denoted as $P_{A}$. {For the circuit power, we set $P_B = 39\mathrm{dBm}$, $P_F = 30\ \mathrm{dBm}$, $P_S = 5\ \mathrm{dBm}$, and $P_A = 20\ \mathrm{dBm}$. The other main parameters are set in Table~\ref{systemParameter2}.}

Fig.~\ref{Convergency} shows the convergence behavior of the proposed algorithm, where the value of objective function (OBJ) at iteration $0$ means the value of OBJ with the initialized feasible parameters. In the case where there are $K=4$ users and the minimum required rate $R_{\mathrm{min}} = 1\mathrm{Mbps}$, our proposed alternating algorithm converges after $3$ rounds. When the number of users increases to $K=8$ and $K=16$, the required iteration rounds decrease to $2$ and $1$, respectively. This is because when the number of users increases, the interference of the users also increases. Hence, the minimum rate is harder to satisfy and the proposed algorithm is more easily to get the locally optimal solution. To prove this, we reduce $R_{\mathrm{min}}$ to $0.8\mathrm{Mbps}$, $0.5\mathrm{Mbps}$, and $0.3\mathrm{Mbps}$. Then, we can observe that with a smaller $R_{\mathrm{min}}$, the larger the number of required iteration rounds, the higher the converged OJB value.

\begin{table}[t]
\caption{{MAIN SYSTEM PARAMETERS}}
\centering
\begin{tabular}{|c|c|c|}
\hline
{\textbf{Parameter} }   & {\textbf{Notation}} & {\textbf{Value}} \\ \hline
{Bandwidth of the system  }  & {$B$ }            & {$10\ \rm{MHz}$   }     \\
{Maximum transmit power }& {$P_\mathrm{max}$}  & {$5\ \rm{dbm}$ }         \\
{Minimum distance between FAs  }             &{ $d_{\mathrm{min}}$ }          & {$\frac{\lambda}{2}$  }           \\
{Minimum required achievable rate }              &{ $R_{\mathrm{min}}$   }        & { $1\mathrm{Mbps}$  }           \\
{Noise power at BS}      & {$\sigma_r^2$  }   & {$-174\  \rm{dbm/Hz}$ }       \\
{Noise power at FAR-U }     & {$\sigma_u^2$ }    & {$-90\  \rm{dbm/Hz}$  }      \\
{The number of FAs of FAR  }   & {$M$   }       & { $4$ }            \\
{The number of FAs at BS }    & {$N$  }        & {$4$ }            \\
{The number of users}     & {$K$ }         & {$4$  }           \\
{Racian factors }    & {$K_0/K_1$}          & {$3+\sqrt{12}$  }           \\
{Path-loss exponent}  & ${l}$          & {$2.6$} \\   {Length/width/height of the blockage } & {$L/W/H$  }        & {$10/0.3/5\ \mathrm{m}$} \\ 
\hline
\end{tabular}
\label{systemParameter2}
\end{table}

\begin{figure}[t]
\centering
\includegraphics[width=1\linewidth]{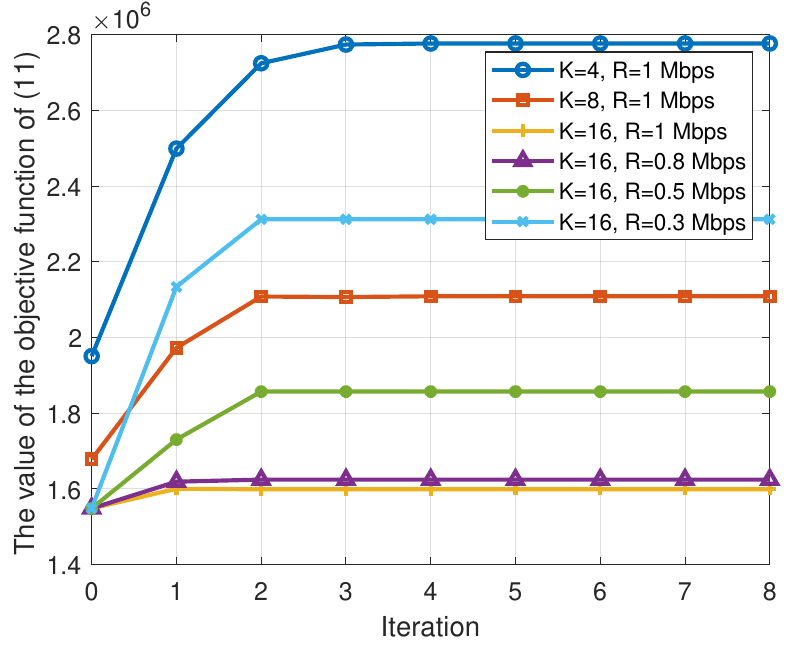}
\caption{Convergency behavior of the proposed algorithm.} 
\label{Convergency}
\end{figure}

Fig.~\ref{Rmin} further demonstrates the relationships between EE and the number of users, and between EE and minimum required rate. In Fig.~\ref{Rmin}, scheme label `$K=k$, FAR' means using the proposed algorithm with $k$ users. As we can observe, the EEs of all three schemes decrease when the minimum required rate increases. This is because with higher minimum rate requirement, the minimum rate constraint is harder to satisfy, and thus the EE becomes lower. Besides, under the same minimum required rate, scheme `$K=4$, FAR' always {holds} the highest EE, since there is the least user interference.

\begin{figure}[t]
\centering
\includegraphics[width=1\linewidth]{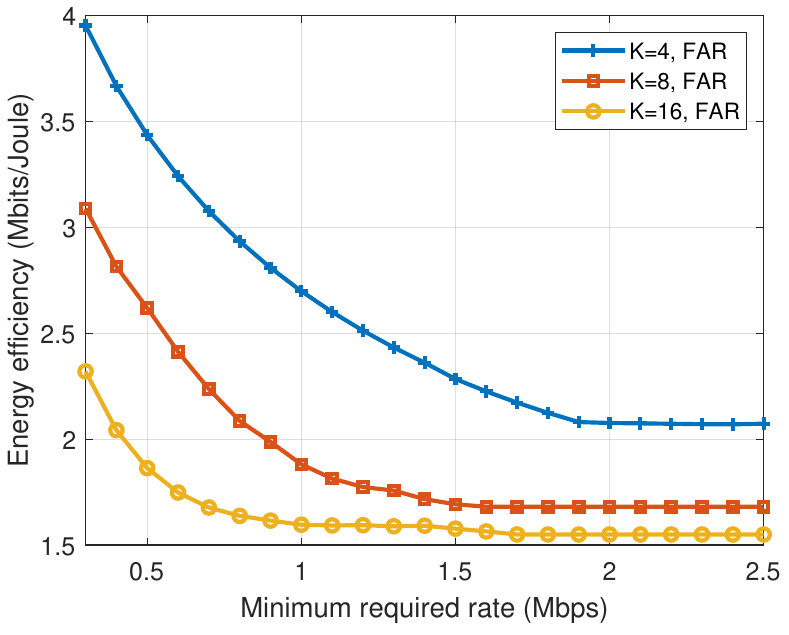}
\caption{{Energy efficiency versus minimum required rate of the system}.} 
\label{Rmin}
\end{figure}

\begin{figure}[t]
\centering
\includegraphics[width=1\linewidth]{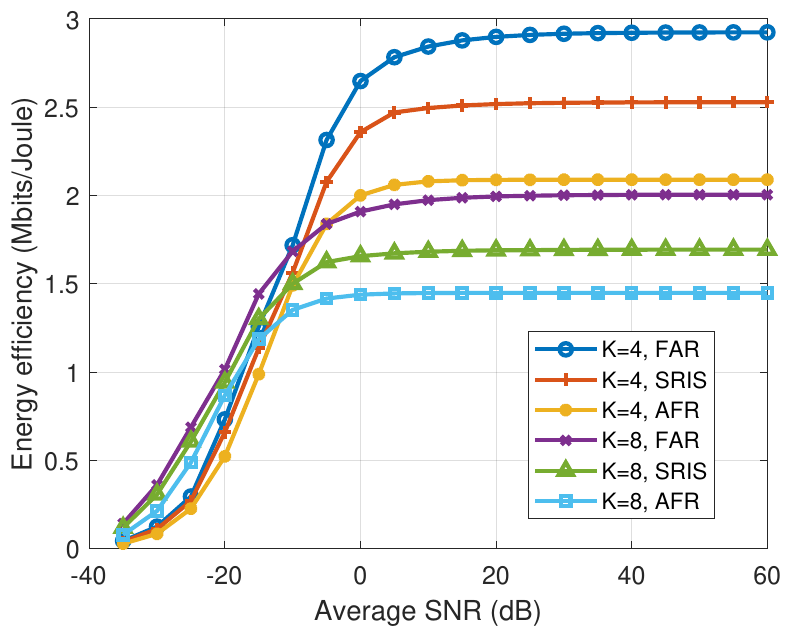}
\caption{{Energy efficiency versus average SNR.}} 
\label{Psigma}
\end{figure}
Fig.~\ref{Psigma} presents the numerical results on how EE changes as the average SNR of proposed system varies, {where the average SNR is defined as $P_{\mathrm{max}}/(\beta_0\beta_k)^{-1}\sigma_r^2$}. For low average SNR, all schemes have very small EEs, since the noise power is much higher than the signal power. When the average SNR becomes higher, the EEs of all schemes first increase and then hold on different levels. From this figure, in general, the smaller the number of users, the larger the EE value that the system can maintain when the SNR is large enough. At the same time, `FAR' always achieves a higher EE than those of the `SRIS' and `AFR' under the same SNR and number of users. In particular, `FAR' can increase up to {$15.65\%$ and $39.94\%$ }EE with $K=4$ users, compared with `SRIS' and `AFR', respectively. Besides, for schemes with $K=8$ users, our proposed algorithm can improve the EE by up to {$18.35\%$ and $38.36\%$} compared to the `SRIS' and `AFR', respectively.

\begin{figure}[t]
\centering
\includegraphics[width=1\linewidth]{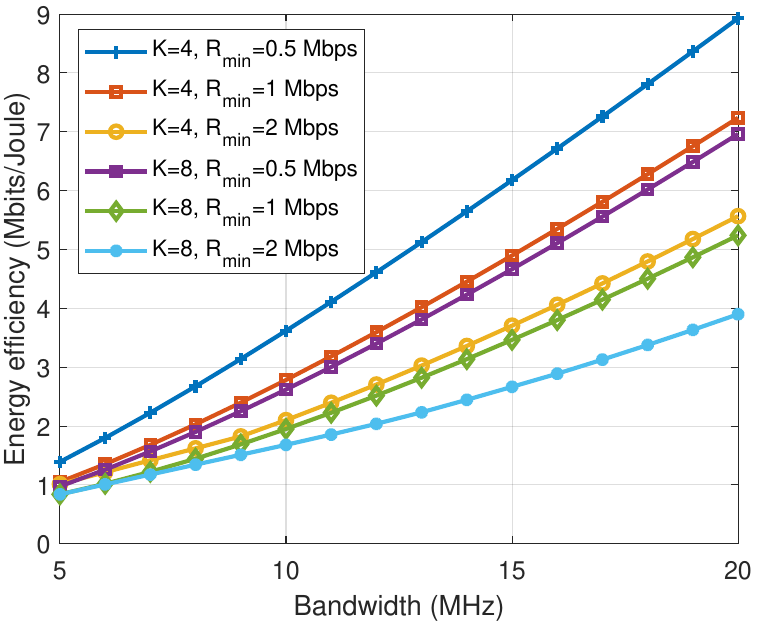}
\caption{Energy efficiency versus bandwidth of the system.} 
\label{Bandwidth}
\end{figure}
To reveal how the system bandwidth affects the system performance, Fig.~\ref{Bandwidth} investigates the EE of the proposed system versus the bandwidth. From the legend, EEs of all schemes increase with the increasing of the bandwidth of the system. For a given bandwidth, fewer users in out proposed system means less interference, and hence, it typically corresponds to a higher EE. However, if the minimum required achievable rate is set too high, an EE of the scheme with fewer users may be lower than the scheme with more users but a smaller $R_{\mathrm{min}}$, as represented by the two schemes $K=4, R_{\mathrm{min}}=2\mathrm{Mbps}$, and $K=8, R_{\mathrm{min}}=0.5\mathrm{Mbps}$ .

\begin{figure}[t]
\centering
\includegraphics[width=1\linewidth]{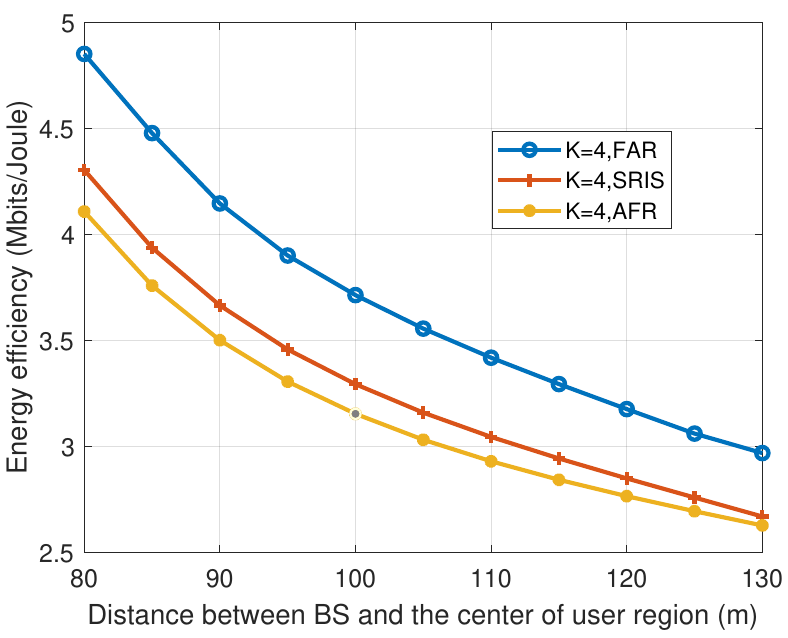}
\caption{{Energy efficiency versus the distance between BS and the center of user region.}} 
\label{Y0}
\end{figure}

{From Fig.~\ref{Convergency}-\ref{Bandwidth}, we can find that the EE performance of `FAR' scheme with different users holds similar changing trends. Hence, to demonstrate the simulation results more clearly, we only illustrate the results simulated with $K=4$ users. Besides, in order to obtain more diverse simulation results, we randomly distribute the users within a wider rectangular area instead of being uniformly distributed in a circle.} {With the redefined user region, we first explore the EE performance versus the distance between the center of the user region and BS. Since the distance between the FAR and BS is set as $Y_0=50\ \mathrm{m}$, we set the distance between user region and BS ranging from $80\ \mathrm{m}$ to $130\ \mathrm{m}$. As demonstrated in Fig.~\ref{Y0}, the EE performance of three schemes all decreases with respect to the increasing distance. This is because a farther distance causes a bigger large-scale fading, decreasing the sum rate of the system. With the same distance, the `FAR' scheme always outperforms the `SRIS' scheme and `AFR' scheme, with up to $13.74\%$ and $19.15\%$ higher EEs, respectively.}

{Fig.~\ref{Psigr} illustrates the EE of the proposed FAR-assisted system versus the noise power at the BS. In particular, we compare the proposed `FAR' scheme with `SRIS' scheme and `AFR' scheme, all with $K=4$ users. When the noise power at the BS ranges from $-170\ \mathrm{dBm}$ to $-70\ \mathrm{dBm}$, all the EEs of three schemes first remain, and when the noise power increases to be higher than $-110\ \mathrm{dBm}$, the EEs of all the schemes decrease. This is because with a higher noise power at BS the noise, the received signal gets a smaller SINR, resulting in smaller EE of the system. Moreover, as the noise power exceeds $-100\ \mathrm{dBm}$ and continues to increase, the EE of both schemes `$K=4, \mathrm{SRIS}$' and `$K=4, \mathrm{AFR}$' rapidly decreases, while the EE performance of them gradually tends to be the same. The EE of our proposed `$K=4, \mathrm{FAR}$' scheme rapidly decreases only after the noise power exceeds $-85\ \mathrm{dBm}$ and continues to increase. The EEs of the three schemes are all at very small values when the noise power was $-60\ \mathrm{dBm}$, which is consistent with the simulation results shown in Fig.~\ref{Psigma}. With the same value of noise power at BS, the proposed `FAR' scheme always outperforms the `SRIS' scheme and the `AFR' scheme with up to $23.39\%$ and $32.25\%$ higher EE.}

\begin{figure}[t]
\centering
\includegraphics[width=1\linewidth]{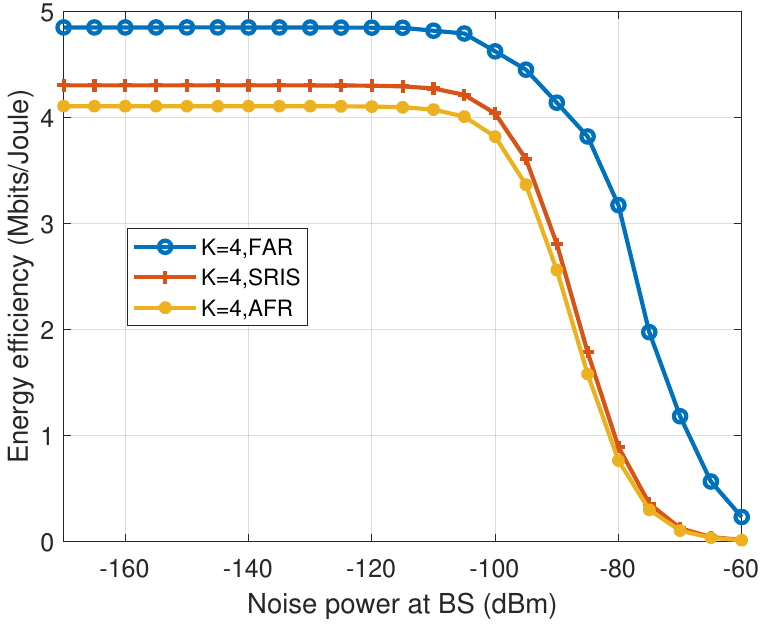}
\caption{{Energy efficiency versus noise power at the BS.}} 
\label{Psigr}
\end{figure}

{Fig.~\ref{FARU} explores the changing trend of EE performance of the proposed FAR-assisted system versus the number of FAs deployed at the FAR-U. In this legend, with given number of the FAs at the FAR-U, we compare the EE performance of schemes with different minimum rate demands. In general, EE performance becomes better with larger number of FAs at the FAR-U. This is because more FAs at the FAR-U expand the spatial dimensions of the received signals at the FAR-U, and further improve the SINR of the system. On the other hand, with the same number of FAs at the FAR-U, there is a higher energy efficiency of the system with a fewer $R_{\mathrm{min}}$, which are consistent with the patterns shown in Figs.~\ref{Rmin} and ~\ref{Bandwidth}.}
\begin{figure}[t]
\centering
\includegraphics[width=1\linewidth]{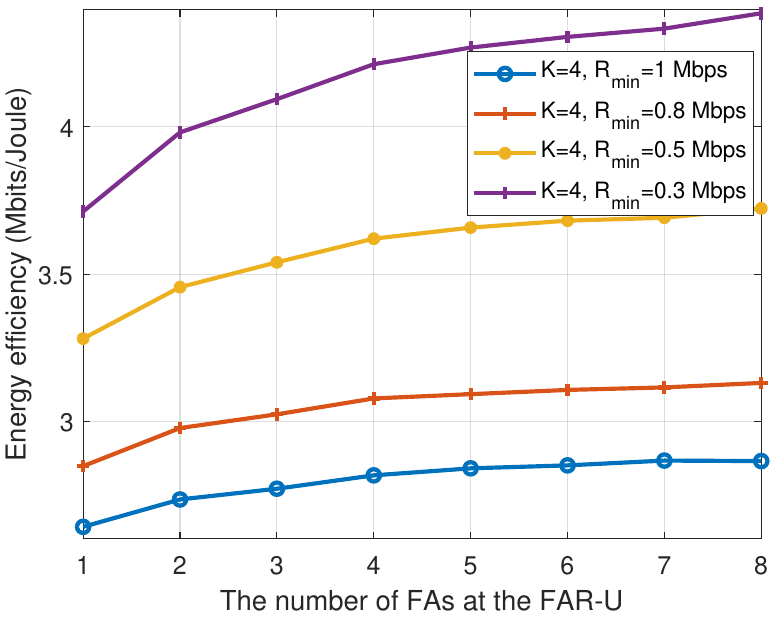}
\caption{{Energy efficiency versus the number of FAs at the FAR-U.}} 
\label{FARU}
\end{figure}

\section{Conclusion}
In this paper, we have investigated the energy efficient FAR-assisted uplink wireless communication system with multiple users. We have first modeled the FAR, {derived} the amplitude attenuation and phase shift of signals transmitting through the blockage with the help of FAR, and formulated the problem with the objective of maximizing system energy efficiency. To solve this problem, we have proposed an iterative algorithm by alternately solving the large-scale fading optimization sub-problem, small-scale fading optimization sub-problem, and joint power control and beamforming design sub-problem. Simulation results have validated the proposed algorithm outperforms the STAR-RIS scheme and the traditional AF relay scheme. The work in this paper bridges the gap in FAS research by focusing on NLoS scenarios and provides a comprehensive framework for FAR-aided wireless communication. 
{The FAR system has great potential to be introduced in indoor communication\cite{luo2017indoor}, covert communication\cite{li2024covert}, green communication\cite{xu2025rate}, {blue}and mobile edge computing (MEC) system\cite{10971386}. Moreover, communication technologies, such as beamforming design\cite{kumar2025beamforming,10689678,10437283}, channel deduction\cite{chen2025channel}, and performance optimization\cite{10196368,10615837,10000992} compatible with the FAR are also interesting exploration directions.}

\begin{appendices}
    \section{Proof of Theorem 1}\label{Proof1}
{Different from} the vector notation stated in Section~\ref{Introduction}, in Appendix~\ref{Proof1}, we use $\vec{A}$ to denote the vector $\vec{A}$.

Given the fact that FAs can be much smaller than the width of an FAR, TE mode plane waves are assumed to be transmitted from FA on the FAR-U to FA on the FAR-B.

If denote $\vec{H}$ is magnetic intensity vector, $\vec{D}$ is electric displacement vector, and $\vec{J}_c$ is conductive current density, we can get the differential form of Ampere's total current law as
\begin{equation}\label{AmpereTotalCurrentLaw}
    \nabla \times \vec{H} =  j \omega\vec{D} + \vec{J}_c.
\end{equation}

Since an isotropic medium is assumed, we can obtain
\begin{equation}\label{Ohm's law}
    \vec{J}_c =  \sigma\vec{E}, \vec{D} = \varepsilon \vec{E},
\end{equation}
{where $\vec{E}$ is the electric intensity field vector.}

Substituting \eqref{Ohm's law} into \eqref{AmpereTotalCurrentLaw}, we can obtain
\begin{align}
    \nabla \times \vec{H} =  j \omega\varepsilon\vec{E} + \sigma\vec{E} 
    = j \omega\underbrace{(\varepsilon-j\frac{\sigma}{\omega})}_{\triangleq \tilde{\varepsilon}}\vec{E},
\end{align}
where $\tilde{\varepsilon}$ is the complex permittivity.

Then, the Maxwell's equations with introducing $\tilde{\varepsilon}$ can be written as
\begin{equation}\label{ME 1}
    \nabla\times \vec{E}  =-j\omega\mu \vec{H},
\end{equation}
\begin{equation}\label{ME 2}
    \nabla\times \vec{H}  =j\omega\tilde{\varepsilon}\vec{E},
\end{equation}
\begin{equation}\label{ME 3}
\nabla\cdot \vec{H}  =0, 
\end{equation}
\begin{equation}\label{ME 4}
\nabla\cdot\vec{E}  =0.
\end{equation}

Take the curl of both sides of \eqref{ME 1}, we can get
\begin{equation}\label{ME Calculation}
        \nabla \times (\nabla\times \vec{E}) =-j\omega\mu \nabla \times \vec{H}.
\end{equation}

Substituting \eqref{ME 4} into the left hand of \eqref{ME Calculation} and substituting \eqref{ME 2} into the right hand of \eqref{ME Calculation}, we have
\begin{equation}\label{solution ME}
        (\nabla^2 + \omega^2\mu \tilde{\varepsilon})\vec{E}= 0,
\end{equation}
where $\vec{E}$ is no longer coupled with $\vec{H}$.

Based on the properties of the TE mode, {the electric field intensity vector at the point along the propagation direction with distance $d$ away from the {$q$-th} FA {at} the FAR-U} can be given as
\begin{equation}\label{solution E}
        \vec{E}(d)= E_{{q}} e^{-j \tilde{\kappa}d} \hat{E}_0,
\end{equation}
where $E_{q}$ is the amplitude of the electric field intensity, $\hat{E}_0$ is the unit vector of electric field intensity, and $\tilde{\kappa} = (\omega^2\mu \tilde{\varepsilon})^{\frac{1}{2}} = \sqrt{\omega^2\mu\varepsilon} (1-j \dfrac{\sigma}{\omega \varepsilon})$.

Then, we have
\begin{align}
\nonumber
    \tilde{\kappa} &= C_1 \left(1+\dfrac{\sigma^2}{\omega^2\varepsilon^2}\right)^{\frac{1}{4}} \left(\dfrac{1}{\sqrt{1+\dfrac{\sigma^2}{\omega^2\varepsilon^2}}}- \dfrac{j\sigma}{\omega \varepsilon\sqrt{1+\dfrac{\sigma^2}{\omega^2\varepsilon^2}}}\right)\\
    \nonumber
    &= C_1\left[1+\mathrm{tan}^2\left(\gamma\right)\right]^{\frac{1}{4}}\left[\mathrm{cos}\left(\gamma\right) - j\mathrm{sin}\left(\gamma\right)\right]^{\frac{1}{2}}\\
    &= C_1 \mathrm{sec}^{\frac{1}{2}}\left(\gamma\right)\left[\mathrm{cos}\left(\frac{\gamma}{2}\right) - j\mathrm{sin}\left(\frac{\gamma}{2}\right)\right].
\end{align}

Therefore, the electric field intensity vector $\vec{E_{p}}$ at the {$p$-th} FA at FAR-B can be given as
\begin{equation}\label{amplitudeAttenuationCoefficient}
\frac{E_p}{E_q}=
    \underbrace{e^{-C_1\mathrm{sec}^{\frac{1}{2}}(\gamma)\mathrm{sin}(\frac{1}{2}\gamma)d_{pq}} }_{\alpha_{pq}}\ \underbrace{e^{-jC_1\mathrm{sec}^{\frac{1}{2}}(\gamma)\mathrm{cos}(\frac{1}{2}\gamma)d_{pq}}}_{e^{-j\theta_{pq}}},
\end{equation}
where $d_{pq}$ is the distance between the {$q$-th} FA {at} the FAR-U and the {$p$-th} FA {at} the FAR-B.

Note that when the reference points of {$\mathcal{C}_u$} and {$\mathcal{C}_b$} are given, $d_{pq}$ can be approximately as $||\mathbf{o}_b-\mathbf{o}_u||_2$. Therefore, we can further have
\begin{equation}
    \alpha_{pq} \approx e^{-C_1\mathrm{sec}^{\frac{1}{2}}(\gamma)\mathrm{cos}(\frac{1}{2}\gamma)||\mathbf{o}_b-\mathbf{o}_u||_2} \triangleq \alpha.
\end{equation}

{\section{Derivations of $\nabla\tilde{h}_k^u(\mathbf{u}_m)$, $\nabla^2\tilde{h}_k^u(\mathbf{u}_m)$,\\ and a feasible $\iota_k^u$}\label{Proof2}}
{
Similar to \eqref{phaseDifference}, the phase difference between the $l$-th FA at the FAR-U and $\mathbf{o}_{u}$ can be given as 
\begin{equation}\label{proof21}
\rho(\mathbf{u}_l) = e^{j \frac{2\pi}{\lambda}(\bm{\kappa}_{u_k})^T (\mathbf{u}_l - \mathbf{o}_{u})}
\end{equation}
where $\bm{\kappa}_{u_k}=[\mathrm{sin}\theta_{u_k}\mathrm{cos}\phi_{u_k},\mathrm{cos}\theta_{u_k},\mathrm{sin}\theta_{u_k}\mathrm{sin}\phi_{u_k}]^T$, $\theta_{u_k}$, and $\phi_{u_k}$ are the wave vector, the elevation of AoA, and the azimuth of AoA of the signal from user $k$, respectively.

Denote $\bm{\Theta}^T=[\bm{\vartheta}_1,\dots,\bm{\vartheta}_M]$, where $\bm{\vartheta}_i=[\theta_{i1},\dots,\theta_{iM}]^T$, and we have $\mathbf{h}_k^u(\mathbf{u}_m)=[\bm{\vartheta}_1^T\mathbf{h}_k,\dots,\bm{\vartheta}_M^T\mathbf{h}_k]^T$. Without loss of generality, we consider the situation where $K_1>>1$, i.e., for the $l$-th entry of $\mathbf{h}_k$, we have $h_k(l)=\rho(\mathbf{u}_l)\rho^*(\mathbf{t}_k)$. For other situations, corresponding parameters can be added to the amplitude and phase of $h_k(l)$ for correction.

Specifically, for the $i$-th entry of $\mathbf{h}_k^u(\mathbf{u}_m)$, we have
\begin{align}\label{proof22}
\nonumber
    \bm{\vartheta}_i^T\mathbf{h}_k &=\sum_{l=1}^{M} e^{-j\theta_{il}}\rho(\mathbf{u}_l)\rho^*(\mathbf{t}_k)\\
    &=\sum_{l=1}^Me^{j\left[\frac{2\pi}{\lambda}\left(\bm{\kappa}_{u_k}^T\tilde{\mathbf{u}}_l-\kappa_{t_k}\right)-\theta_{il}\right]},
\end{align}
where $\tilde{\mathbf{u}}_l=\mathbf{u}_l-\mathbf{o}_u$, and $\kappa_{t_k}=\bm{\kappa}_{t_k}^T(\mathbf{t}_k-\mathbf{o}_k)$.

For the sake of simplicity in writing, we denote 
\begin{equation}\label{proof23}
    \mathbf{g}=\mathbf{h}_k^{u}(\mathbf{u}_m^{(n)})^H\hat{\mathbf{H}}_0=[|g_1|e^{j\angle g_1},\dots,|g_M|e^{j\angle g_M}]^T,
\end{equation}
where $|g_i|$ and $\angle g_i$ are the amplitude and phase of the $i$-th entry of $\mathbf{g}$, respectively.

Based on \eqref{proof22} and \eqref{proof23}, $\tilde{h}_k^u(\mathbf{u}_m)$ can be further given as 
\begin{align}\label{proof24}
\nonumber
   \tilde{h}_k^u(\mathbf{u}_m) &= \mathrm{Re}\left[\mathbf{g}^H\mathbf{h}_k^u(\mathbf{u}_m)\right]\\
   \nonumber
   &=\mathrm{Re}\left\{\sum_{i=1}^M|g_i|\sum_{l=1}^M e^{j\left[\frac{2\pi}{\lambda}\left(\bm{\kappa}_{u_k}^T\tilde{\mathbf{u}}_l-\kappa_{t_k}\right)-\theta_{il}-\angle g_i\right]}\right\}\\
   &= \sum_{i=1}^M|g_i|\sum_{l=1}^M\mathrm{cos}\left[\frac{2\pi}{\lambda}\left(\bm{\kappa}_{u_k}^T\tilde{\mathbf{u}}_l-\kappa_{t_k}\right)-\theta_{il}-\angle g_i\right].
\end{align}

Then, the gradient of $\tilde{h}_k^u(\mathbf{u}_m)$ over $\mathbf{u}_m$ is given as
\begin{align}\label{proof25}
\nonumber
    &\frac{\partial \tilde{h}_k^u(\mathbf{u}_m)}{\partial \mathbf{u}_m}=\\
    &-\sum_{i=1}^M|g_i|\mathrm{sin}\left(\angle \tilde{u}_m\right)\left(\frac{2\pi}{\lambda}\bm{\kappa}_{u_k}+\hat{C}_1\frac{\mathbf{b}_i-\mathbf{u}_m}{||\mathbf{b}_i-\mathbf{u}_m||_2}\right),
\end{align}
and the Hessian matrix of $\tilde{h}_k^u(\mathbf{u}_m)$ over $\mathbf{u}_m$ is given as
\begin{align}\label{proof26}
    \nonumber
    &\frac{\partial^2 \tilde{h}_k^u(\mathbf{u}_m)}{\partial \mathbf{u}_m \partial \mathbf{u}_m^T}=-\sum_{i=1}^M|g_i|\Bigg\{\mathrm{cos}(\angle \tilde{u}_m)\Big[\frac{4\pi^2}{\lambda^2}\bm{\kappa}_{u_k}\bm{\kappa}_{u_k}^T\\
    \nonumber
    &+\frac{2\pi\hat{C_1}[\bm{\kappa}_{u_k}(\mathbf{b}_i-\mathbf{u}_m)^T+(\mathbf{b}_i-\mathbf{u}_m)\bm{\kappa}_{u_k}^T]}{\lambda||\mathbf{b}_i-\mathbf{u}_m||_2}\\
    \nonumber
    &+\hat{C}_1^2\frac{(\mathbf{b}_i-\mathbf{u}_m)(\mathbf{b}_i-\mathbf{u}_m)^T}{||\mathbf{b}_i-\mathbf{u}_m||_2^2}\Big]\\
    &+\hat{C}_1\mathrm{sin}(\angle \tilde{u}_m)\Big[\frac{-\mathbf{I}_3}{||\mathbf{b}_i-\mathbf{u}_m||_2}+\frac{(\mathbf{b}_i-\mathbf{u}_m)(\mathbf{b}_i-\mathbf{u}_m)^T}{||\mathbf{b}_i-\mathbf{u}_m||_2^3}\Big]\Bigg\},
\end{align}
where $\angle \tilde{u}_m = \frac{2\pi}{\lambda}\left(\bm{\kappa}_{u_k}^T\tilde{\mathbf{u}}_m-\kappa_{t_k}\right)-\theta_{im}-\angle g_i$ and $\hat{C}_1 = C_1\mathrm{sec}^{\frac{1}{2}}(\gamma)\mathrm{cos}(\frac{1}{2}\gamma)$.

To find a $\iota_k^u$ satisfying $\iota_k^u\mathbf{I}_3-\nabla^2\tilde{h}_k^u(\mathbf{u}_m)\succeq0$, we can select a $\iota_k^u$, where $||\nabla^2\tilde{h}_k^u(\mathbf{u}_m)||_2 \leq \iota_k^u $, since $||\nabla^2\tilde{h}_k^u(\mathbf{u}_m)||_2\mathbf{I}_3-\nabla^2\tilde{h}_k^u(\mathbf{u}_m)\succeq0$. Moreover, for $ ||\nabla^2\tilde{h}_k^u(\mathbf{u}_m)||_2$, we have
\begin{align}\label{proof27}
 \nonumber
 ||\nabla^2\tilde{h}_k^u(\mathbf{u}_m)||_2^2
 &\leq  ||\nabla^2\tilde{h}_k^u(\mathbf{u}_m)||_F^2\\
 &=\sum_{i=1}^3\sum_{l=1}^3|\nabla^2\tilde{h}_k^u(\mathbf{u}_m)|_{i,l}^2,
\end{align}
where $|\nabla^2\tilde{h}_k^u(\mathbf{u}_m)|_{i,l}$ is the modulus of the entry in the $i$-th row and the $l$-th column of $\nabla^2\tilde{h}_k^u(\mathbf{u}_m)$.

Denote $T=\mathrm{max}\{L,W,H\}$, and we have $l_{b_i}-l_{u_m}\leq T, \forall l=x,y,z$. Besides, due to the facts that $||\mathbf{b}_i-\mathbf{u}_m||_2\geq W$, $\hat{C}_1 \leq C_1$, and all trigonometric functions have values no more than $1$, for any entry of $\nabla^2\tilde{h}_k^u(\mathbf{u}_m)$, we have
\begin{align}\label{proof28}
    \nonumber
    &|\nabla^2\tilde{h}_k^u(\mathbf{u}_m)|_{i,l}\leq\\
    & \sum_{i=1}^M |g_i|\underbrace{\left(\frac{4\pi^2}{\lambda^2}+\frac{4\pi C_1T}{\lambda W}+\frac{C_1^2T^2}{W^2}+\frac{C_1}{W}+\frac{C_1T^2}{W^3}\right)}_{\triangleq C_2}.
\end{align}


Considering that $y_{u_m}=Y_o+W$ is a constant, the entries in Hessian matrix $\nabla^2\tilde{h}_k^u(\mathbf{u}_m)$ that obtain partial derivatives with respect to $y_{u_m}$ are all equal 0. Hence, we can set $\iota_k^u$ as
\begin{equation}\label{proof29}
\iota_k^u=2\sum_{i=1}^M|g_i|C_2\geq||\nabla^2\tilde{h}_k^u(\mathbf{u}_m)||_F\geq||\nabla^2\tilde{h}_k^u(\mathbf{u}_m)||_2,
\end{equation}
which further satisfies $\iota_k^u\mathbf{I}_3-\nabla^2\tilde{h}_k^u(\mathbf{u}_m)\succeq0$.
}

 {
 \section{A global lower bound of $\tilde{h}_k^u(\mathbf{u}_m)$ }\label{proof3}
 According to Taylor’s theorem, for a given point $\mathbf{u}_m^{(n)}$, we can always have $\xi$, satisfying

\begin{align}\label{proof31}
\nonumber
 \tilde{h}_k^u(\mathbf{u}_m)
 &=  \tilde{h}_k^u(\mathbf{u}_m^{(n)})+\nabla\tilde{h}_k^u(\mathbf{u}_m^{(n)})^T\left(\mathbf{u}_m-\mathbf{u}_m^{(n)}\right)\\
 &+\frac{1}{2}\left(\mathbf{u}_m-\mathbf{u}_m^{(n)}\right)^T\nabla^2\tilde{h}_k^u(\xi)\left(\mathbf{u}_m-\mathbf{u}_m^{(n)}\right).
\end{align}

Since the spectral norm is always non-negative, we know that $||\nabla^2\tilde{h}_k^u(\mathbf{u}_m)||_2 \geq -\iota_k^u$. Further, for any $\mathbf{u}_m$, we have
\begin{equation}\label{proof32}
    -\iota_k^u\mathbf{I}_3-\nabla^2\tilde{h}_k^u(\mathbf{u}_m)\preceq0.
\end{equation}
Hence, we have
\begin{align}\label{proof33}
\nonumber
 &\tilde{h}_k^u(\mathbf{u}_m)
 =  \tilde{h}_k^u(\mathbf{u}_m^{(n)})+\nabla\tilde{h}_k^u(\mathbf{u}_m^{(n)})^T\left(\mathbf{u}_m-\mathbf{u}_m^{(n)}\right)\\
 \nonumber
 &+\frac{1}{2}\left(\mathbf{u}_m-\mathbf{u}_m^{(n)}\right)^T\nabla^2\tilde{h}_k^u(\xi)\left(\mathbf{u}_m-\mathbf{u}_m^{(n)}\right)\\
 \nonumber
 &\geq \tilde{h}_k^u(\mathbf{u}_m^{(n)})+\nabla\tilde{h}_k^u(\mathbf{u}_m^{(n)})^T\left(\mathbf{u}_m-\mathbf{u}_m^{(n)}\right)\\
 \nonumber
 &-\frac{1}{2}\left(\mathbf{u}_m-\mathbf{u}_m^{(n)}\right)^T\iota_k^u\mathbf{I}_3\left(\mathbf{u}_m-\mathbf{u}_m^{(n)}\right)\\
  \nonumber
 &=-\frac{1}{2}\iota_k^u \mathbf{u}_m^T\mathbf{u}_m+\left(\nabla\tilde{h}_k^u(\mathbf{u}_m^{(n)})+\iota_k^u\mathbf{u}_m^{(n)}\right)^T\mathbf{u}_m\\
 &-\left(\nabla\tilde{h}_k^u(\mathbf{u}_m^{(n)})+\frac{1}{2}\iota_k^u\mathbf{u}_m^{(n)}\right)^T\mathbf{u}_m^{(n)}+\tilde{h}_k^u(\mathbf{u}_m^{(n)}),
\end{align}
from which we can obtain a lower bound of $\tilde{h}_k^u(\mathbf{u}_m)$.
}

{
\section{A global upper bound of $h_i^u(\mathbf{u}_m)$}\label{proof4}
Denote $\bm{\Theta}=[\bm{\theta_1},\dots,\bm{\theta_M}]$, where $\bm{\theta}_i=[\theta_{1i},\dots,\theta_{Mi}]^T$, and we can rewrite $|\bm{\omega}_k^H\mathbf{H}_0 \bm{\Theta}\mathbf{h}_i|^2$ as shown in \eqref{proof41}, at the top of the next page, where $h_i(m)$ is the $m$-th entry of $\mathbf{h}_i$.

\begin{figure*}
{
    \begin{align}\label{proof41}
    \nonumber
  h_k^i(\mathbf{u}_m)
  &=  
   \mathrm{tr}\left[h_i(m)h_i^*(m)\bm{\theta}_m\bm{\theta}_m^H\hat{\mathbf{H}}_0\right]+ \mathrm{tr}\left[h_i(m)\bm{\theta}_m\sum_{l=1,l\neq i}^Mh_i^*(l)\bm{\theta}_l^H\hat{\mathbf{H}}_0\right]\\
   &+\mathrm{tr}\left[\sum_{l=1,l\neq i}^Mh_i(l)\bm{\theta}_l h_i^*(m)\bm{\theta}_m^H \hat{\mathbf{H}}_0\right] + \mathrm{tr}\left[\sum_{l=1,l\neq i}^Mh_i(l)\bm{\theta}_l\sum_{p=1,p\neq i}^M h_i^*(p)\bm{\theta}_p^H \hat{\mathbf{H}}_0\right],
    \end{align}
}
\hrulefill

\end{figure*}

Based on \eqref{proof41}, we can further have the following equality and inequalities as
\begin{align}\label{proof42}
     \nonumber
     h_i^u(\mathbf{u}_m)&\overset{(\mathrm{a})}{=}
    \mathrm{tr}\left[h_i(m)h_i^*(m)\bm{\theta}_m\bm{\theta}_m^H\hat{\mathbf{H}}_0\right]+ 2\underbrace{\mathrm{Re}\left(\bm{\theta}_m^H\mathbf{c}_i^u\right)}_{\triangleq h_i^u(\mathbf{u}_m)}+C_3\\
    \nonumber
    &\overset{(\mathrm{b})}{\leq}
    M\lambda_{\mathrm{max}}^{u,i}+ 2 h_i^u(\mathbf{u}_m)+\mathrm{Re}(C_3)\\
    \nonumber
    &\overset{(\mathrm{c})}{\leq}
    M\lambda_{\mathrm{max}}^{u,i}+\mathrm{Re}(C_3)+2h_i^u(\mathbf{u}_m^{(n)})+\iota_i^u\mathbf{u}_m^T\mathbf{u}_m\\
    \nonumber
    &+2\left(\nabla h_i^u(\mathbf{u}_m^{(n)})-\iota_i^u\mathbf{u}_m^{(n)}\right)^T\mathbf{u}_m\\
    &+\left(2\nabla{h}_i^u(\mathbf{u}_m^{(n)})+\iota_i^u\mathbf{u}_m^{(n)}\right)^T\mathbf{u}_m^{(n)}\triangleq\tilde{\tilde{h}}_i^u(\mathbf{u}_m),
\end{align}
where $\mathbf{c}_i^u=\hat{\mathbf{H}}_0\sum_{l=1,l\neq i}^Mh_i(l)\bm{\theta}_l h_i^*(m)$ is a constant vector, $C_3=\mathrm{tr}\left[\sum_{l=1,l\neq i}^Mh_i(l)\bm{\theta}_l\sum_{p=1,p\neq i}^M h_i^*(p)\bm{\theta}_p^H \hat{\mathbf{H}}_0\right]$ is a constant scalar, and $\iota_k^i$ is a positive real number making $\iota_i^u\mathbf{I}_3-\nabla^2 h_i^u(\mathbf{u}_m)\succeq0$.

The equality $(\mathrm{a})$ holds because $\hat{\mathbf{H}}_0$ is a Hermitian matrix, and then the second term and the third term of \eqref{proof41} can be denoted as $  \mathrm{tr}\left[(\mathbf{c}_i^u)^{H}\bm{\theta}_m\right]$ and $\mathrm{tr}\left[\bm{\theta}_m^H\mathbf{c}_i^u\right]$, respectively, the sum of which is $2\mathrm{Re}\left(\bm{\theta}_m^H\mathbf{c}_i^u\right)$. The inequality $(\mathrm{b})$ holds because Lemma~\ref{lemma3} provides a upper bound of the first item of $h_i^u(\mathbf{u}_m)$, and Lemma~\ref{lemma4} proves the inequality marked by $(\mathrm{c})$.

\begin{lemma}\label{lemma3}
    Since $h_i(m)h_i^*(m)\hat{\mathbf{H}}_0$ is a Hermitian matrix, we denote its largest eigenvalue is $\lambda_{\mathrm{max}}^{u,i}$, and we can get that
    \begin{equation}\label{proof43}
        \mathrm{tr}\left[h_i(m)h_i^*(m)\bm{\theta}_m\bm{\theta}_m^H\hat{\mathbf{H}}_0\right] \leq M\lambda_{\mathrm{max}}^{u,i}.
    \end{equation}
\end{lemma}
\begin{proof}
    See Appendix~\ref{proof5}.
\end{proof}
\begin{lemma}\label{lemma4}
    Denote the $l$-th entry of $\mathbf{c}_i^u$ is $|c_i^u(l)|e^{j\angle c_i^u(l)}$, and $C_4=\frac{C_1(T^2+W^2)}{W^3}+\frac{C_1^2T^2}{W^2}$. Then, select $\iota_i^u = 2C_4\sum_{l=1}^M|c_i^u(l)|$, we have a global bound of $h_i^u(\mathbf{u}_m)$ as 
    \begin{align}\label{proof44}
    \nonumber
        &h_i^u(\mathbf{u}_m) \leq \frac{1}{2}\iota_i^u\mathbf{u}_m^T\mathbf{u}_m+\left(\nabla h_i^u(\mathbf{u}_m^{(n)})-\iota_i^u\mathbf{u}_m^{(n)}\right)^T\mathbf{u}_m\\
    &+\left(\frac{1}{2}\iota_i^u\mathbf{u}_m^{(n)}-\nabla{h}_i^u(\mathbf{u}_m^{(n)})\right)^T\mathbf{u}_m^{(n)}+h_i^u(\mathbf{u}_m^{(n)}).
    \end{align}
\end{lemma}
\begin{proof}
    See Appendix~\ref{proof6}.
\end{proof}

\section{Proof of Lemma~\ref{lemma3}}\label{proof5}
Denote $\mathbf{\hat{H}}_0^m=h_i(m)h_i^*(m)\hat{\mathbf{H}}_0$, and the left hand side term of \eqref{proof43} can be rewritten as 
\begin{equation}
    \mathrm{tr}\left[h_i(m)h_i^*(m)\bm{\theta}_m\bm{\theta}_m^H\hat{\mathbf{H}}_0\right]=\bm{\theta}_m^H\mathbf{\hat{H}}_0^m\bm{\theta}_m.
\end{equation}
Based on Rayleigh-Ritz theorem, for any $\bm{\theta}_m\neq \mathbf{0}$, we have
\begin{equation}
\bm{\theta}_m^H\mathbf{\hat{H}}_0^m\bm{\theta}_m \leq \lambda_{\mathrm{max}}^{u,i}\bm{\theta}_m^H\bm{\theta}_m=M\lambda_{\mathrm{max}}^{u,i}.
\end{equation}

\section{Proof of Lemma~\ref{lemma4}}\label{proof6}

Similar to \eqref{proof31} in Appendix~\ref{proof3}, at a given point $\mathbf{u}_m^{(n)}$, there is $\xi$ satisfying
\begin{align}\label{proof61}
  \nonumber  h_i^u(\mathbf{u}_m)&=h_i^u(\mathbf{u}_m^{(n)})+\nabla h_i^u(\mathbf{u}_m^{(n)})^T\left(\mathbf{u}_m-\mathbf{u}_m^{(n)}\right)\\
 &+\frac{1}{2}\left(\mathbf{u}_m-\mathbf{u}_m^{(n)}\right)^T\nabla^2h_i^u(\xi)\left(\mathbf{u}_m-\mathbf{u}_m^{(n)}\right),
\end{align}

With a $\iota_i^u\geq ||\nabla^2h_i^u(\mathbf{u}_m)||_2$, we have
\begin{equation}\label{proof62}
(\mathbf{u}_m-\mathbf{u}_m^{(n)})^T(\iota_i^u\mathbf{I}_3-\nabla^2 h_i^u(\xi))(\mathbf{u}_m-\mathbf{u}_m^{(n)})\succeq0.
\end{equation}

Then, substituting $\nabla^2 h_i^u(\xi)$ in \eqref{proof61} by $\iota_i^u\mathbf{I}_3$, we can get the relationship shown in \eqref{proof44}.

To determine such a $\iota_i^k$, we first derive the Hessian matrix of $h_i^u(\mathbf{u}_m)$, and find a positive real number greater than its Frobenius norm. 

With the definition about $\mathbf{c}_i^u$ in Lemma~\ref{lemma4}, we have 
\begin{align}
    \nonumber h_i^u(\mathbf{u}_m)&=\mathrm{Re}\left\{\sum_{l=1}^M|c_i^u(l)|e^{j[\theta_{lm}+\angle c_i^u(l)]}\right\}\\
    &=\sum_{l=1}^M|c_i^u(l)|\mathrm{cos}[\theta_{lm}+\angle c_i^u(l)].
\end{align}
The gradient and the Hessian matrix of it can be respectively given as
\begin{equation}
    \frac{\partial h_i^u(\mathbf{u}_m)}{\partial \mathbf{u}_m}=\sum_{l=1}^M \frac{\hat{C}_1|c_i^u(l)|\mathrm{sin}[\theta_{lm}+\angle c_i^u(l)](\mathbf{b}_l-\mathbf{u}_m)}{||\mathbf{b}_l-\mathbf{u}_m||_2},
\end{equation}
and
\begin{align}
\nonumber
    &\frac{\partial^2 h_i^u(\mathbf{u}_m)}{\partial \mathbf{u}_m\partial \mathbf{u}_m^T}=\sum_{l=1}^M|c_i^u(l)|\Bigg\{\hat{C}_1\mathrm{sin}[\theta_{lm}+\angle c_i^u(l)]\\
    \nonumber
    &\Big[\frac{-\mathbf{I}_3}{||\mathbf{b}_i-\mathbf{u}_m||_2}+\frac{(\mathbf{b}_i-\mathbf{u}_m)(\mathbf{b}_i-\mathbf{u}_m)^T}{||\mathbf{b}_i-\mathbf{u}_m||_2^3}\Big]\\
    &-\hat{C}_1^2\mathrm{cos}[\theta_{lm}+\angle c_i^u(l)]\frac{(\mathbf{b}_l-\mathbf{u}_m)(\mathbf{b}_l-\mathbf{u}_m)^T}{||\mathbf{b}_l-\mathbf{u}_m||_2^2}\Bigg\}.
\end{align}

Similarly to the inequality in \eqref{proof28}, we have
\begin{align}
    \nonumber
    &||\nabla^2h_i^u(\mathbf{u}_m)||_2^2 \leq ||\nabla^2h_i^u(\mathbf{u}_m)||_F^2\\
    &\leq4\left[\underbrace{\left(\frac{C_1(T^2+W^2)}{W^3}+\frac{C_1^2T^2}{W^2}\right)}_{\triangleq C_4}\sum_{l=1}^M|c_i^u(l)|\right]^2 .
\end{align}
Select $\iota_i^u=2C_4\sum_{l=1}^M|c_i^u(l)|$, where $\iota_i^u>||\nabla^2h_i^u(\mathbf{u}_m)||_2$, i.e., $\iota_i^u\mathbf{I}_3-\nabla^2h_i^u(\mathbf{u}_m)\succeq0$.
}
\end{appendices}
\bibliographystyle{IEEEtran}
\bibliography{main}

\begin{thebibliography}{10}
\providecommand{\url}[1]{#1}
\csname url@samestyle\endcsname
\providecommand{\newblock}{\relax}
\providecommand{\bibinfo}[2]{#2}
\providecommand{\BIBentrySTDinterwordspacing}{\spaceskip=0pt\relax}
\providecommand{\BIBentryALTinterwordstretchfactor}{4}
\providecommand{\BIBentryALTinterwordspacing}{\spaceskip=\fontdimen2\font plus
\BIBentryALTinterwordstretchfactor\fontdimen3\font minus \fontdimen4\font\relax}
\providecommand{\BIBforeignlanguage}[2]{{%
\expandafter\ifx\csname l@#1\endcsname\relax
\typeout{** WARNING: IEEEtran.bst: No hyphenation pattern has been}%
\typeout{** loaded for the language `#1'. Using the pattern for}%
\typeout{** the default language instead.}%
\else
\language=\csname l@#1\endcsname
\fi
#2}}
\providecommand{\BIBdecl}{\relax}
\BIBdecl

\bibitem{wong2020fluid}
K.-K. Wong, A.~Shojaeifard, K.-F. Tong, and Y.~Zhang, ``Fluid antenna systems,'' \emph{IEEE Transactions on Wireless Communications}, vol.~20, no.~3, pp. 1950--1962, 2020.

\bibitem{wong2020performance}
K.~K. Wong, A.~Shojaeifard, K.-F. Tong, and Y.~Zhang, ``Performance limits of fluid antenna systems,'' \emph{IEEE Communications Letters}, vol.~24, no.~11, pp. 2469--2472, 2020.

\bibitem{chai2022port}
Z.~Chai, K.-K. Wong, K.-F. Tong, Y.~Chen, and Y.~Zhang, ``Port selection for fluid antenna systems,'' \emph{IEEE Communications Letters}, vol.~26, no.~5, pp. 1180--1184, 2022.

\bibitem{psomas2023continuous}
C.~Psomas, P.~J. Smith, H.~A. Suraweera, and I.~Krikidis, ``Continuous fluid antenna systems: Modeling and analysis,'' \emph{IEEE Communications Letters}, 2023.

\bibitem{abu2021liquid}
H.~Abu~Bakar, R.~Abd~Rahim, P.~J. Soh, and P.~Akkaraekthalin, ``Liquid-based reconfigurable antenna technology: Recent developments, challenges and future,'' \emph{Sensors}, vol.~21, no.~3, p. 827, 2021.

\bibitem{zhang2024pixel}
J.~Zhang, J.~Rao, Z.~Ming, Z.~Li, C.-Y. Chiu, K.-K. Wong, K.-F. Tong, and R.~Murch, ``A pixel-based reconfigurable antenna design for fluid antenna systems,'' \emph{arXiv preprint arXiv:2406.05499}, 2024.

\bibitem{ghadi2024physical}
F.~R. Ghadi, K.-K. Wong, F.~J. L{\'o}pez-Mart{\'\i}nez, W.~K. New, H.~Xu, and C.-B. Chae, ``Physical layer security over fluid antenna systems: Secrecy performance analysis,'' \emph{IEEE Transactions on Wireless Communications}, 2024.

\bibitem{wang2024fluid}
C.~Wang, G.~Li, H.~Zhang, K.-K. Wong, Z.~Li, D.~W.~K. Ng, and C.-B. Chae, ``Fluid antenna system liberating multiuser mimo for isac via deep reinforcement learning,'' \emph{IEEE Transactions on Wireless Communications}, 2024.

\bibitem{wang2024ai}
C.~Wang, Z.~Li, K.-K. Wong, R.~Murch, C.-B. Chae, and S.~Jin, ``Ai-empowered fluid antenna systems: Opportunities, challenges, and future directions,'' \emph{IEEE Wireless Communications}, 2024.

\bibitem{wong2021fluid}
K.-K. Wong and K.-F. Tong, ``Fluid antenna multiple access,'' \emph{IEEE Transactions on Wireless Communications}, vol.~21, no.~7, pp. 4801--4815, 2021.

\bibitem{skouroumounis2022fluid}
C.~Skouroumounis and I.~Krikidis, ``Fluid antenna with linear mmse channel estimation for large-scale cellular networks,'' \emph{IEEE Transactions on Communications}, vol.~71, no.~2, pp. 1112--1125, 2022.

\bibitem{zhou2024fluid}
L.~Zhou, J.~Yao, M.~Jin, T.~Wu, and K.-K. Wong, ``Fluid antenna-assisted isac systems,'' \emph{IEEE Wireless Communications Letters}, 2024.

\bibitem{ye2023fluid}
Y.~Ye, L.~You, J.~Wang, H.~Xu, K.-K. Wong, and X.~Gao, ``Fluid antenna-assisted mimo transmission exploiting statistical csi,'' \emph{IEEE Communications Letters}, 2023.

\bibitem{10767351}
Y.~Chen, M.~Chen, H.~Xu, Z.~Yang, K.-K. Wong, and Z.~Zhang, ``Joint beamforming and antenna design for near-field fluid antenna system,'' \emph{IEEE Wireless Communications Letters}, vol.~14, no.~2, pp. 415--419, 2025.

\bibitem{10092780}
B.~Tang, H.~Xu, K.-K. Wong, K.-F. Tong, Y.~Zhang, and C.-B. Chae, ``Fluid antenna enabling secret communications,'' \emph{IEEE Communications Letters}, vol.~27, no.~6, pp. 1491--1495, 2023.

\bibitem{gan2021ris}
X.~Gan, C.~Zhong, C.~Huang, and Z.~Zhang, ``Ris-assisted multi-user miso communications exploiting statistical csi,'' \emph{IEEE Transactions on Communications}, vol.~69, no.~10, pp. 6781--6792, 2021.

\bibitem{magbool2024surveyintegratedsensingcommunication}
\BIBentryALTinterwordspacing
A.~Magbool, V.~Kumar, Q.~Wu, M.~D. Renzo, and M.~F. Flanagan, ``A survey on integrated sensing and communication with intelligent metasurfaces: Trends, challenges, and opportunities,'' 2024. [Online]. Available: \url{https://arxiv.org/abs/2401.15562}
\BIBentrySTDinterwordspacing

\bibitem{10720877}
A.~Magbool, V.~Kumar, A.~Bazzi, M.~F. Flanagan, and M.~Chafii, ``Multi-functional ris for a multi-functional system: Integrating sensing, communication, and wireless power transfer,'' \emph{IEEE Network}, vol.~39, no.~1, pp. 71--79, 2025.

\bibitem{10133841}
M.~Ahmed, A.~Wahid, S.~S. Laique, W.~U. Khan, A.~Ihsan, F.~Xu, S.~Chatzinotas, and Z.~Han, ``A survey on star-ris: Use cases, recent advances, and future research challenges,'' \emph{IEEE Internet of Things Journal}, vol.~10, no.~16, pp. 14\,689--14\,711, 2023.

\bibitem{10380743}
X.~Li, Z.~Tian, W.~He, G.~Chen, M.~C. Gursoy, S.~Mumtaz, and A.~Nallanathan, ``Covert communication of star-ris aided noma networks,'' \emph{IEEE Transactions on Vehicular Technology}, vol.~73, no.~6, pp. 9055--9060, 2024.

\bibitem{10685065}
F.~Yu, C.~Zhang, and T.~Q.~S. Quek, ``Star-ris-enabled simultaneous indoor-and-outdoor communication networks: A stochastic geometry approach,'' \emph{IEEE Transactions on Wireless Communications}, vol.~23, no.~12, pp. 18\,053--18\,069, 2024.

\bibitem{abdou2024sum}
S.~B.~S. Abdou, W.~K. New, C.~Y. Leow, S.~Won, K.-K. Wong, and Z.~Ding, ``Sum-rate maximization for uav relay-aided fluid antenna system with noma,'' in \emph{2024 IEEE 7th International Symposium on Telecommunication Technologies (ISTT)}.\hskip 1em plus 0.5em minus 0.4em\relax IEEE, 2024, pp. 53--58.

\bibitem{10740058}
J.~Zhang, J.~Rao, Z.~Li, Z.~Ming, C.-Y. Chiu, K.-K. Wong, K.-F. Tong, and R.~Murch, ``A novel pixel-based reconfigurable antenna applied in fluid antenna systems with high switching speed,'' \emph{IEEE Open Journal of Antennas and Propagation}, vol.~6, no.~1, pp. 212--228, 2025.

\bibitem{9539785}
Y.~Shen, K.-F. Tong, and K.-K. Wong, ``Reconfigurable surface wave fluid antenna for spatial mimo applications,'' in \emph{2021 IEEE-APS Topical Conference on Antennas and Propagation in Wireless Communications (APWC)}, 2021, pp. 150--152.

\bibitem{9977471}
H.~Wang, Y.~Shen, K.-F. Tong, and K.-K. Wong, ``Continuous electrowetting surface-wave fluid antenna for mobile communications,'' in \emph{TENCON 2022 - 2022 IEEE Region 10 Conference (TENCON)}, 2022, pp. 1--3.

\bibitem{7950976}
P.~Lotfi, S.~Soltani, and R.~D. Murch, ``Printed endfire beam-steerable pixel antenna,'' \emph{IEEE Transactions on Antennas and Propagation}, vol.~65, no.~8, pp. 3913--3923, 2017.

\bibitem{10907789}
X.~Li, J.~Zhao, G.~Chen, W.~Hao, D.~B.~d. Costa, A.~Nallanathan, H.~Shin, and C.~Yuen, ``Star-ris-assisted covert wireless communications with randomly distributed blockages,'' \emph{IEEE Transactions on Wireless Communications}, vol.~24, no.~6, pp. 4690--4705, 2025.

\bibitem{11059266}
L.~Pang, Y.~Wang, Y.~Zhang, Y.~Chen, A.~Wang, and J.~Li, ``Coverage optimization of hybrid active-passive star-ris assisted indoor-outdoor communication under hardware impairments and imperfect csi,'' \emph{IEEE Internet of Things Journal}, pp. 1--1, 2025.

\bibitem{10930784}
Y.~M. Park, S.~S. Hassan, Y.~K. Tun, E.-N. Huh, W.~Saad, and C.~S. Hong, ``Design optimization of noma aided multi-star-ris for indoor environments: A convex approximation imitated reinforcement learning approach,'' \emph{IEEE Transactions on Mobile Computing}, pp. 1--17, 2025.

\bibitem{10707271}
H.~Chi, K.~Cao, H.~Ding, L.~Lv, J.~Chen, D.~Diao, B.~Wang, and F.~Gong, ``Performance analysis for star-ris-assisted wireless powered communications with cooperative jamming,'' \emph{IEEE Internet of Things Journal}, vol.~12, no.~3, pp. 2574--2591, 2025.

\bibitem{ma2023mimo}
W.~Ma, L.~Zhu, and R.~Zhang, ``Mimo capacity characterization for movable antenna systems,'' \emph{IEEE Transactions on Wireless Communications}, 2023.

\bibitem{zhu2023modeling}
L.~Zhu, W.~Ma, and R.~Zhang, ``Modeling and performance analysis for movable antenna enabled wireless communications,'' \emph{IEEE Transactions on Wireless Communications}, 2023.

\bibitem{laourine2008capacity}
A.~Laourine, M.-S. Alouini, S.~Affes, and A.~St{\'e}phenne, ``On the capacity of generalized-k fading channels,'' \emph{IEEE Transactions on Wireless Communications}, vol.~7, no.~7, pp. 2441--2445, 2008.

\bibitem{yang2015role}
A.~Yang, Z.~He, C.~Xing, Z.~Fei, and J.~Kuang, ``The role of large-scale fading in uplink massive mimo systems,'' \emph{IEEE Transactions on Vehicular Technology}, vol.~65, no.~1, pp. 477--483, 2015.

\bibitem{9461768}
Z.~Yang, M.~Chen, W.~Saad, and M.~Shikh-Bahaei, ``Optimization of rate allocation and power control for rate splitting multiple access (rsma),'' \emph{IEEE Transactions on Communications}, vol.~69, no.~9, pp. 5988--6002, 2021.

\bibitem{10032275}
Z.~Yang, M.~Chen, Z.~Zhang, and C.~Huang, ``Energy efficient semantic communication over wireless networks with rate splitting,'' \emph{IEEE Journal on Selected Areas in Communications}, vol.~41, no.~5, pp. 1484--1495, 2023.

\bibitem{dinkelbach1967nonlinear}
W.~Dinkelbach, ``On nonlinear fractional programming,'' \emph{Management science}, vol.~13, no.~7, pp. 492--498, 1967.

\bibitem{9497709}
Z.~Yang, M.~Chen, W.~Saad, W.~Xu, M.~Shikh-Bahaei, H.~V. Poor, and S.~Cui, ``Energy-efficient wireless communications with distributed reconfigurable intelligent surfaces,'' \emph{IEEE Transactions on Wireless Communications}, vol.~21, no.~1, pp. 665--679, 2022.

\bibitem{lobo1998applications}
M.~S. Lobo, L.~Vandenberghe, S.~Boyd, and H.~Lebret, ``Applications of second-order cone programming,'' \emph{Linear algebra and its applications}, vol. 284, no. 1-3, pp. 193--228, 1998.

\bibitem{mu2021simultaneously}
X.~Mu, Y.~Liu, L.~Guo, J.~Lin, and R.~Schober, ``Simultaneously transmitting and reflecting (star) ris aided wireless communications,'' \emph{IEEE Transactions on Wireless Communications}, vol.~21, no.~5, pp. 3083--3098, 2021.

\bibitem{6671453}
A.~Zappone, P.~Cao, and E.~A. Jorswieck, ``Energy efficiency optimization in relay-assisted mimo systems with perfect and statistical csi,'' \emph{IEEE Transactions on Signal Processing}, vol.~62, no.~2, pp. 443--457, 2014.

\bibitem{luo2017indoor}
J.~Luo, L.~Fan, and H.~Li, ``Indoor positioning systems based on visible light communication: State of the art,'' \emph{IEEE communications surveys \& tutorials}, vol.~19, no.~4, pp. 2871--2893, 2017.

\bibitem{li2024covert}
X.~Li, Z.~Tian, W.~He, G.~Chen, M.~C. Gursoy, S.~Mumtaz, and A.~Nallanathan, ``Covert communication of star-ris aided noma networks,'' \emph{IEEE Transactions on Vehicular Technology}, vol.~73, no.~6, pp. 9055--9060, 2024.

\bibitem{xu2025rate}
R.~Xu, Z.~Yang, Y.~Mao, C.~Huang, Q.~Yang, L.~Xu, W.~Xu, and Z.~Zhang, ``Rate-splitting multiple access enabled green probabilistic semantic communication over wireless networks,'' \emph{IEEE Transactions on Green Communications and Networking}, 2025.

\bibitem{10971386}
Y.~Cang, M.~Chen, and Z.~Yang, ``Movable antenna enhanced mec systems with discrete antenna position selection,'' \emph{IEEE Communications Letters}, pp. 1--1, 2025.

\bibitem{kumar2025beamforming}
V.~Kumar and M.~Chafii, ``Beamforming design for secure ris-enabled isac: Passive ris vs. active ris,'' \emph{IEEE Transactions on Wireless Communications}, 2025.

\bibitem{10689678}
A.~Magbool, V.~Kumar, and M.~F. Flanagan, ``Robust beamforming design for fairness-aware energy efficiency maximization in ris-assisted mmwave communications,'' \emph{IEEE Transactions on Communications}, vol.~73, no.~4, pp. 2648--2662, 2025.

\bibitem{10437283}
V.~Kumar, M.~Chafii, A.~L. Swindlehurst, L.-N. Tran, and M.~F. Flanagan, ``Sca-based beamforming optimization for irs-enabled secure integrated sensing and communication,'' in \emph{GLOBECOM 2023 - 2023 IEEE Global Communications Conference}, 2023, pp. 5992--5997.

\bibitem{chen2025channel}
Z.~Chen, Z.~Zhang, Z.~Yang, C.~Huang, and M.~Debbah, ``Channel deduction: A new learning framework to acquire channel from outdated samples and coarse estimate,'' \emph{IEEE Journal on Selected Areas in Communications}, 2025.

\bibitem{10196368}
S.~Zhang, Z.~Yang, M.~Chen, D.~Liu, K.-K. Wong, and H.~V. Poor, ``Beamforming design for the performance optimization of intelligent reflecting surface assisted multicast mimo networks,'' \emph{IEEE Transactions on Wireless Communications}, vol.~23, no.~3, pp. 2325--2339, 2024.

\bibitem{10615837}
S.~Zhang, M.~Chen, W.~Zhang, Z.~Yang, D.~Liu, Z.~Zhang, and K.-K. Wong, ``Performance optimization for multicast mm wave mimo networks with mobile users,'' in \emph{2024 IEEE International Conference on Communications Workshops (ICC Workshops)}, 2024, pp. 846--851.

\bibitem{10000992}
S.~Zhang, Z.~Yang, M.~Chen, D.~Liu, K.-K. Wong, and H.~V. Poor, ``Performance optimization for intelligent reflecting surface assisted multicast mimo networks,'' in \emph{GLOBECOM 2022 - 2022 IEEE Global Communications Conference}, 2022, pp. 5838--5843.

\end{thebibliography}

\end{document}